\documentclass[journal]{IEEEtran}

\usepackage{wrapfig}
\let\labelindent\relax
\usepackage{enumitem,soul}
\usepackage{rotating}
\usepackage{chngcntr}
\usepackage{apptools}
\usepackage{listings}
\AtAppendix{\counterwithin{thm}{section}}
\usepackage{graphicx}
\usepackage{graphics}
\usepackage{amssymb}
\usepackage{amsmath}
\usepackage{color}
\usepackage{soul}
\usepackage{xspace}
\usepackage{algpseudocode}
\usepackage{bbm}
\usepackage{comment}
\usepackage{algorithmicx}
\usepackage{subfig}
\usepackage{psfrag}
\usepackage{float}
\usepackage{cite}
\usepackage{tikz,pgfplots}
\usetikzlibrary{shapes,arrows}
\usetikzlibrary{decorations.pathreplacing,decorations.markings}
\usetikzlibrary{positioning,calc}

\tikzstyle{block}=[draw opacity=0.7,line width=1.4cm]
\definecolor{CranJ}{cmyk}{0,0.69,0.54,0.04} 
\definecolor{PinkJ}{cmyk}{0,0.71,0.43,0.12} 
\definecolor{Cran}{cmyk}{0,0.73,0.41,0.29} 
\definecolor{VRed}{cmyk}{0,0.75,0.25,0.2} 
\definecolor{ORed}{cmyk}{0,0.75,0.75,0} 
\definecolor{CBlue}{cmyk}{1,0.25,0,0} 

\tikzset{block/.style={%
        inner xsep=1mm,
        inner ysep=1.5mm,
        rectangle,very thick,draw}}
\tikzset{sum/.style={%
        circle,
        minimum size=2mm,inner xsep=1.2mm,inner ysep=1.2mm,
        very thick,draw}}
\tikzset{point/.style={%
        minimum size=0mm,inner xsep=4mm,inner ysep=0mm,draw}}
\tikzset{link/.style={->,very thick,>=stealth}}
\tikzset{undirlink/.style={<->,very thick,>=stealth}}
\tikzset{pole/.style={cross out, draw=black, minimum size=2*(#1-\pgflinewidth), inner sep=0pt, outer sep=0pt},cross/.default={1pt}}

\usepackage{rotating}
\usepackage[]{graphicx}
\usepackage{sidecap}

\title{\LARGE \bf
Privacy preservation in continuous-time average consensus algorithm via deterministic additive obfuscation signals}

\author{Navid Rezazadeh and Solmaz S. Kia
  \thanks{The authors are with the Department of Mechanical and Aerospace Engineering, University of California Irvine, Irvine, CA 92697,  
    {\tt\small \{nrezazad,solmaz\}@uci.edu}. This work is supported by NSF award ECCS-1653838.} %
}

\newcommand{\VV}{\mathcal{V}}
\newcommand{\EE}{\mathcal{E}}
\newcommand{\GG}{\mathcal{G}}

\newcommand{\LL}{\vectsf{L}}
\newcommand{\lL}{\vectsf{L}}
\newcommand{\rR}{\vect{\mathsf{R}}}

\newcommand{\real}{{\mathbb{R}}} \newcommand{\reals}{{\mathbb{R}}}
 
\newcommand{\realpositive}{{\mathbb{R}}_{>0}}

\newcommand{\eps}{\epsilon}

\newcommand{\rank}{\operatorname{rank}}

\newcommand{\din}{\mathsf{d}_{\operatorname{in}}}
\newcommand{\Nout}{\mathcal{N}_{\operatorname{out}}}

\newcommand{\dout}{\mathsf{d}_{\operatorname{out}}}
\newcommand{\Dout}{\vectsf{D}^{\operatorname{out}}}

\newcommand{\until}[1]{\{1,\dots,#1\}}

\renewcommand{\baselinestretch}{1.1}

\newcommand{\vect}[1]{\boldsymbol{\mathbf{#1}}}
\newcommand{\vectsf}[1]{\vect{\mathsf{#1}}}

\newcommand{\dvect}[1]{\dot{\vect{#1}}}

\newcommand{\Diag}[1]{\operatorname{Diag}(#1)}

\newcommand{\avrg}[1]{\frac{1}{N}\sum_{j=1}^N#1^j}

 \newcommand{\boxend}{\hfill \ensuremath{\Box}}

\newtheorem{thm}{Theorem}[section]
\newtheorem{prop}{Proposition}[section]
\newtheorem{rem}{Remark}[section]
\newtheorem{cor}{Corollary}[section]
\newtheorem{lem}{Lemma}[section]
\newtheorem{defn}{Definition}

\newcommand{\oprocendsymbol}{\hbox{$\bullet$}}
\newcommand{\oprocend}{\relax\ifmmode\else\unskip\hfill\fi\oprocendsymbol}

\makeatletter
\renewcommand*{\@opargbegintheorem}[3]{\trivlist
      \item[\hskip \labelsep{ #1\ #2}] (#3):\ \itshape}
\makeatother

\usepackage{hyperref} 
\usepackage{xr-hyper} 
\begin{document}
\maketitle
\thispagestyle{empty}
\pagestyle{empty}
\begin{abstract}
  This paper considers the problem of privacy preservation against passive internal and external eavesdroppers in the continuous-time Laplacian average consensus algorithm over strongly connected and weight-balanced digraphs. For this problem, we evaluate the effectiveness of the use of additive obfuscation signals as a privacy preservation measure against eavesdroppers that know the graph topology. Our results include (a) identifying the necessary and sufficient conditions on admissible additive obfuscation signals that do not perturb the convergence point of the algorithm from the average of initial values of the agents; (b) obtaining the necessary and sufficient condition on the knowledge set of an eavesdropper that enables it to identify the initial value of another agent; (c) designing observers that internal and external eavesdroppers can use to identify the initial conditions of another agent when their knowledge set on that agent enables them to do so. We demonstrate our results through a numerical~example.
\end{abstract}

\section{Introduction}\label{sec::Intro}
Decentralized multi-agent cooperative operations have been emerging as effective solutions for some of today's important socio-economical challenges. However, in some areas~involving sensitive data, for example in smart grid, banking or~healthcare applications, the adaption of these solutions is~hindered by concerns over the privacy preservation guarantees of the participating clients. Motivated by the demand~for privacy preservation evaluations and design of privacy-preserving augmentations for existing decentralized solutions, in this~paper we consider the privacy preservation problem~in the distributed~static average consensus problem using additive obfuscation~signals. 
 
The static average consensus problem in a network of agents each endowed with a local static reference value consists of designing a distributed algorithm that enables each agent to asymptotically obtain the average of the static reference values across the network.
The solutions to this problem have been used in various distributed computing, synchronization, and estimation problems as well as control of multi-agent cyber-physical systems. The average consensus problem has been studied extensively in the literature (see e.g.,~\cite{ROS-RMM:04,WR-RWB:05,LX-SB:04},~\cite{ROS-JAF-RMM:07}). The widely adopted distributed solution for the static average consensus problem is the simple first-order Laplacian algorithm in which each agent initializes its local dynamics with its local reference value and transmits this local value to its neighboring agents. Therefore, the reference value is readily revealed to the outside world, and thus the privacy of the agents implementing this algorithm is trivially breached. This paper studies the multi-agent static average consensus problem under the privacy preservation requirement against internal and external passive eavesdroppers in the network. By passive, we mean agents that only listen to the communication messages and want to obtain the reference value of the other agents without interrupting the distributed operation. The solution we examine is to induce privacy preservation property by adding obfuscation signals to the internal dynamics and the transmitted output of~the~agents.

\emph{Literature review}: 
Privacy preservation solutions for the average consensus problem have been investigated in the literature mainly in the context of discrete-time consensus algorithms over connected undirected graphs. The general idea is to add obfuscation signals to the transmitted out signal of the agents. For example, in one of the early privacy-preserving schemes, Kefayati, Talebi, and Khalaj \cite{kefayati2007secure} proposed that each agent adds a random number generated by zero-mean Gaussian processes to its initial condition. This way the reference value of the agents is guaranteed to stay private but the algorithm does not necessarily converge to the anticipated value. Similarly, in recent years, Nozari, Tallapragada and Cortes~\cite{nozari2017differentially} also relied on adding zero-mean noises to protect the privacy of the agents. However, they develop their noises according to a framework defined based on the concept of differential privacy, which is initially developed in the data science literature~\cite{mcsherry2007mechanism,friedman2010data,dwork2008differential,dwork2014algorithmic}. In this framework,~\cite{nozari2017differentially} characterizes the convergence degradation and proposes an optimal noise in order to keep a level of privacy to the agents while minimizing the rate of convergence deterioration. To eliminate deviation from desired convergence point, Manitara and Hadjicostis~\cite{NEM-CNH:13}
proposed to add a zero-sum finite sequence of noises to the transmitted signal of each agent, and  Mo and Murray \cite{mo2017privacy} proposed to add zero-sum infinite sequences. Because of
the zero-sum condition on the obfuscation signals, however, ~\cite{NEM-CNH:13} and~\cite{mo2017privacy} show that the privacy of an agent can only be preserved when the eavesdropper does not have access to at least one of the signals transmitted to that agent. Additive noises have also been used as a privacy preservation mechanism in other distributed algorithms such as distributed optimization~\cite{huang2015differentially} and distributed estimation~\cite{le2012differentially,JLN-GJP:14}. A thorough
review of these results can be found in a recent tutorial paper~\cite{JC-GED-SH-JLN-SM-GJP:16}.
For the discrete-time average consensus, on a different approach,~\cite{ruan2017secure} uses a cryptographic approach to preserve the privacy of the~agents. Moreover,~\cite{ASE-SSK:20} proposes to use the dynamic average consensus algorithm of~\cite{SSK-JC-SM:15-ijrnc} as a privacy-preserving algorithm for the average consensus problem.

\emph{Statement of contributions:} We consider the problem of privacy preservation of the continuous-time static Laplacian average consensus algorithm over strongly connected and weight-balanced digraphs using additive obfuscation signals. 
Similar to the reviewed literature above, in our privacy preservation analysis, 
we consider the extreme case that the eavesdroppers know the graph topology. But, instead of stochastic obfuscations, here we use deterministic obfuscations signals. These obfuscations are in the form of continuous-time integrable signals that we add to the transmitted out signal of the agents are also to the agreement dynamics of the agents. We refer to the obfuscation signals that do not disturb the final convergence point of the algorithm as \emph{admissible obfuscation signals}.
In our approach, instead of using by the customary zero-sum vanishing additive admissible signals, we start by carefully examining the stability and convergence proprieties of the static average consensus algorithm in the presence of the obfuscations to find the necessary and sufficient conditions on the admissible obfuscation signals. The motivation is to explore whether there exist other types of admissible obfuscation signals that can extend the privacy preservation guarantees. An interesting theoretical finding of our study is that the admissible obfuscation signals do not have to be vanishing. Also, we show that the necessary and sufficient conditions that specify the admissible obfuscation signals of the agents are highly coupled. We discuss how the agents can choose their admissible obfuscation signals locally with or without coordination among themselves. The conditions we obtain to define the locally chosen admissible obfuscation signals are coupled through a set of under-determined linear algebraic constraints with constant scalar free variables. 

Understanding the nature of the admissible obfuscation signals is crucial in the privacy preservation evaluations, as it is rational to assume that the eavesdroppers are aware of the necessary conditions on such~signals and use them to breach the privacy of the agents. In our study, we evaluate the privacy preservation of the Laplacian average consensus algorithm with additive obfuscation signals against internal and external eavesdroppers, depending on whether the coupling variables of the necessary conditions defining the locally chosen admissible obfuscation signals are known to the eavesdropper or not. This way, we study privacy preservation against the most informed eavesdroppers and also explore what kind of guarantees we can provide against less informed eavesdroppers that do not know some parameters. We show that when the coupling variables are known to the eavesdroppers, they can use this extra piece of information to enhance their knowledge set to discover the private value of the other agents. In this case, our main result states that the necessary and sufficient condition for an eavesdropper to be able to identify the initial value of another agent is to have direct access to all the signals transmitted to and out of the agent. When this condition is not satisfied, the privacy guarantee is that the eavesdropper not only cannot obtain the exact reference value but also cannot establish an estimate on it. Precisely, to show that any agent $i$ is private, we show that across the network there are arbitrarily different reference values, including for agent $i$, for which the signals received by the eavesdropper is exactly the same as those corresponding to the initializing the algorithm at the actual reference values. This shows that the use of deterministic obfuscation signals results in a stronger privacy guarantee than the stochastic approaches such as $\eps$-differential privacy~\cite{nozari2017differentially} and of~\cite{mo2017privacy}  where even though the exact reference value is concealed, an estimate on the reference value can be obtained, see, e.g.,  \cite[Fig. 4]{mo2017privacy}.

Our next contribution is to design asymptotic observers that internal and external eavesdroppers that have access to all the input and output signals of an agent can use to identify that agent's initial condition. For these observers, we also characterize the time history of their estimation error. Our results show that external eavesdroppers need to use an observer with a higher numerical complexity to compensate for the local state information that internal eavesdroppers can use. As another contribution, we identify examples of graph topologies in which the privacy of all the agents is preserved using additive admissible obfuscation signals. On the other hand, if the coupling variables of the necessary conditions defining the locally chosen admissible obfuscation signals are unknown to the eavesdroppers, we show that the eavesdroppers cannot reconstruct the private reference value of the other agents even if they have full access to all the transmitted input and output signals of an agent. We use input-to-state stability (ISS) results~(see~\cite{EDS:06,SND-DVE-EDS:11}) to perform our analysis. 

A preliminary version of our work has appeared in~\cite{NR-SSK:18acc}. In
this paper the results are extended in the following directions:
(a) we derive the necessary and sufficient conditions to characterize the admissible signals;
(b) we study privacy preservation also with respect to external eavesdroppers;
(c) we consider a general class of a set of measurable essentially bounded  obfuscation signals;
(d) we improve our main result from sufficient condition to necessary and sufficient~condition.

\section{Preliminaries}

We denote the standard Euclidean norm of vector
$\vect{x}\in\reals^n$ by $\|\vect{x}\|=\sqrt{\vect{x}^\top\vect{x}}$, and 
the (essential) supremum norm  of a signal $f:\real_{\geq0}\to\real^n$ by $\|f\|_{\textup{ess}}=(\textup{ess})\sup\{\|f(t)\|,t\geq0\}$. The set of measurable essentially bounded functions $f:\real_{\geq0}\to\real^n$ is denoted by $\mathcal{L}^{\infty}_n$. The set of measurable functions $f:\real_{\geq0}\to\real^n$ that satisfy $\int_{0}^t\|f(\tau)\|\textup{d}\tau<\infty$ is denoted by  $\mathcal{L}^{1}_n$. For sets $\mathcal{A}$ and $\mathcal{B}$, the relative complement of $\mathcal{B}$ in $\mathcal{A}$ is $\mathcal{A}\backslash\mathcal{B}=\{x\in\mathcal{A}\,|\,x\not\in\mathcal{B}\}$. For a vector $\vect{x}\in\real^n$, the sum of its elements is $\texttt{sum}(\vect{x})$.
In~a network of $N$ agents, to emphasize that a variable is local to an agent $i\in\until{N}$, we use superscripts. Moreover, if ${p}^i\in\reals$ is a variable of agent $i\in\until{N}$, the aggregated ${p}^i$'s of the network is the vector~$\vect{p} =[\{{p}^i\}_{i=1}^N]=
[{{p}^1},\cdots,{{p}^N}]^\top \in \reals^N$. 

\emph{Graph theory}: we review some basic concepts from algebraic graph
theory following~\cite{FB-JC-SM:09}.  
A weighted directed graph (digraph) is a triplet $\GG = (\VV ,\EE,
\vect{\sf{A}})$, where
$\VV=\{1,\dots,N\}$ is the \emph{node set}, $\EE \subseteq \VV\times \VV$ is the \emph{edge set} and $
\vect{\sf{A}}=[\mathsf{a}_{ij}]\in\real^{N\times N}$ is a weighted \emph{adjacency}
matrix with the property that $ \mathsf{a}_{ij} >0$ if $(i, j) \in\EE$
and $\mathsf{a}_{ij} = 0$, otherwise. A weighted digraph is
\emph{undirected} if $\mathsf{a}_{ij} = \mathsf{a}_{ji}$ for all
$i,j\in\VV$.  An edge from $i$ to $j$,
denoted by $(i,j)$, means that agent $j$ can send information to agent
$i$. For an edge $(i,j) \in\EE$, $i$ is called an \emph{in-neighbor}
of $j$ and $j$ is called an \emph{out-neighbor} of~$i$.  We denote the
set of the out-neighbors of an agent $i\in\VV$ by $\mathcal{N}_{\textup{out}}^i$.  
We define $\mathcal{N}_{\textup{out}+i}^i=\mathcal{N}_{\textup{out}}^i\cup \{i\}$.  A \emph{directed path} is a sequence of nodes connected by edges.  A
digraph is called \emph{strongly connected} if for every pair of
vertices there is a directed path connecting~them. We refer to a strongly connected and undirected graph as
a \emph{connected graph}. The \emph{weighted out-degree} and
\emph{weighted in-degree} of a node $i$, are respectively, $\din^i
=\sum^N_{j =1} \mathsf{a}_{ji}$ and $\dout^i =\sum^N_{j =1}
\mathsf{a}_{ij}$.
A digraph is \emph{weight-balanced} if at each node $i\in\VV$, the
weighted out-degree and weighted in-degree coincide (although they
might be different across different nodes).  The \emph{(out-)
  Laplacian} matrix is $\lL=[\ell_{ij}]$ is $\lL= \vect{\mathsf{D}}^{\textup{out}} -
\vect{\mathsf{A}}$, where $\vect{\mathsf{D}}^{\textup{out}} =
\Diag{\dout^1,\cdots, \dout^N} \in \reals^{N \times N}$.  Note that
$\lL\vect{1}_N=\vect{0}$. A digraph is weight-balanced if and only if
$\vect{1}_N^\top\lL=\vect{0}$. For a strongly connected and weight-balanced digraph, $\rank(\lL)=N-1$, $\rank(\lL+\lL^\top)=N-1$, and $\lL$ has one zero eigenvalue $\lambda_1=0$ and the rest of its eigenvalues have positive real parts. We let  $\rR\in\real^{N\times(N-1)}$ be a matrix whose columns are normalized orthogonal complement of $\vect{1}_N$.~Then 
\begin{align}\label{eq::lL+}
    \vect{T}^\top\lL\vect{T}\!=\!\begin{bmatrix}0&\vect{0}\\
    0&\lL^{+}\end{bmatrix},~~ \vect{T}\!=\!\begin{bmatrix}\frac{1}{\sqrt{N}}\vect{1}_N&\rR\end{bmatrix},~~ \lL^{+}\!=\!\rR^\top\lL\rR.
\end{align}
For a strongly connected and weight-balanced digraph, $-\lL^{+}$ is a Hurwitz matrix.

\section{Problem formulation}
Consider the static average consensus algorithm
\begin{align}\label{eq::consensus}&\dot{x}^i(t)=-\sum\nolimits_{j=1}^N\!\!\mathsf{a}_{ij} \,(x^i(t)-x^j(t)),\quad x^i(0)=\mathsf{r}^i,
\end{align}
over a strongly connected and weight-balanced digraph $\mathcal{G}(\VV,\EE,\vectsf{A})$. For such an interaction typology,  $x^i$ of each agent $i\in\VV$ converges to $\avrg{\mathsf{r}}$ as $t\to\infty$~\cite{ROS-JAF-RMM:07}. In this algorithm, $\mathsf{r}^i$, represents a \emph{reference value} of agent $i\in\VV$.
Because in~\eqref{eq::consensus}, 
the reference value $\mathsf{r}^i$ of each agent $i\in\VV$ is transmitted to its in-neighbors, this algorithm  trivially reveals the reference value $\mathsf{r}^i$ of each agent $i\in\mathcal{V}$ to all its in-neighbors and any external agent that is listening to the communication messages. In this paper, we investigate whether in a network of $N\geq 3$ agents, the reference value of the agents can be concealed from the \emph{eavesdroppers} by adding the obfuscation signals $f^i\in\mathcal{L}^{\infty}_{1}$  and $g^i\in\mathcal{L}^{\infty}_{1}$ to, respectively,~the internal dynamics and the transmitted signal of each agent $i\in\VV$ (see Fig.~\ref{fig::network_island}), i.e.,   
\begin{subequations}\label{eq::consensus-modified}
\begin{align}
\dot{x}^i(t)&=-\sum\nolimits_{j=1}^N\mathsf{a}_{ij} \,(x^i(t)-y^j(t))+f^i(t),\label{eq::consensus-modified-x}\\
y^i(t)&=x^i(t)+g^i(t),
\label{eq::consensus-modified-y}\\
x^i(0)&=\mathsf{r}^i,\end{align}
\end{subequations}
while still guaranteeing that $x^i$ converges to $\avrg{\mathsf{r}}$ as $t\to\infty$. We refer to the set of obfuscation signals $\{f^j,g^j\}_{j=1}^N$ in~\eqref{eq::consensus-modified} for~which each agent $i\!\in\!\VV$ still converges to the  average~of~the~reference values across the network, i.e.,  $\lim_{t\to\infty}x^i(t)\!=\!\avrg{x}(0)\!=\!\avrg{\mathsf{r}}$,
as the \emph{admissible obfuscation~signals}.
We define the eavesdroppers formally as follows.

\begin{defn}[eavesdropper]\label{def::eavesdropper} An eavesdropper is an agent inside (internal agent) or outside (external agent) the network that stores and processes the accessible inter-agent communication messages to obtain the private reference value of the other agents in the network, without interfering with the execution of algorithm~\eqref{eq::consensus-modified}. 
\end{defn}

\begin{defn}[Privacy preservation]\label{df::privacy}
Consider an eavesdropper as defined in Definition~\ref{def::eavesdropper}, that has access to $y^j(t)$, $t\in\real_{\geq0}$, of all agents $j \in \mathcal{O} \subset \mathcal{V}$ in a network that implements~\eqref{eq::consensus-modified} with locally chosen admissible perturbation signals $(f^l,g^l)$, $l\in\VV$. We say the privacy of an agent $i \in \mathcal{V}$ is preserved if for any arbitrary $\gamma \in \real_{>0}$, there exists a tuple $\{{x^i}'(0)={\mathsf{r}^i}', {f^i}'(t), {g^i}'(t)\}$, with locally chosen admissible perturbations $({f^i}'(t), {g^i}'(t))$ and $\left|{\mathsf{r}^i}'-\mathsf{r}^i\right| > \gamma$, such that $y^j(t) \equiv {y^j}'(t) $, $t\in\real_{\geq0}$, for all $j \in \mathcal{O}$.
\end{defn}

When privacy of an agent $i\in\VV$ is preserved in accordance to Definition~\ref{df::privacy}, it means that there exists arbitrary number of execution of algorithm~\eqref{eq::consensus-modified} with arbitrary different reference values ${\mathsf{r}^i}'$ ($\left|{\mathsf{r}^i}'-\mathsf{r}^i\right| > \gamma$ for any $\gamma\in\real_{>0}$) for agent $i$  for which the signals received by the eavesdropper in all the executions are identical. Privacy preservation according to Definition~\ref{df::privacy} is  stronger than the privacy preservation in stochastic approaches such as~\cite{mo2017privacy}, where even though the exact reference value is concealed, an estimate with a quantifiable confidence interval on the reference value can be obtained; see Section~\ref{sec::numeric} for more~discussion.

We examine the privacy preservation properties of algorithm~\eqref{eq::consensus-modified} against non-collaborative eavesdroppers. The eavesdroppers are non-collaborative if they do not share their \emph{knowledge sets} with each other. The knowledge set of an eavesdropper is the information that it can use to infer the private reference value of the other agents. The extension of our results to collaborative agents is rather straightforward and is omitted for~brevity. Without loss of generality, we assume that agent $1$ is the internal eavesdropper that wants to obtain the reference value of other agents in the network. At each time $t\in\real_{\geq0}$, the signals that  are available to agent $1$ are $$\mathcal{Y}^1(t)=\{x^1(\tau),y^1(\tau),\{y^i(\tau)\}_{i\in\Nout^1}\}_{\tau=0}^t.$$ For an external eavesdropper, the available signals depend on what channels it intercepts. We assume that the external eavesdropper can associate the intercepted signals to the corresponding agents. We represent the set of these signals with $\mathcal{Y}^{ext}(t)$.
We assume that the eavesdropper knows the graph topology. It is also rational to assume that the eavesdroppers are aware of the form of the necessary conditions on the admissible obfuscation signals.

{
\begin{figure}\centering
\begin{tikzpicture}[align=center,photon/.style={decorate,decoration={snake,post length=1mm}},scale=0.75, every node/.style={scale=0.75}]
  \newcommand{\x}{0.6cm};
  \newcommand{\y}{0.8cm};
  \node (u1)  {$\{y^j(t)\}_{j \in \mathcal{N}_{\textup{out}}^i}$};
  \node (hpre) at ($(u1.east)+(5*\x,-0.1)$) [block] {$-\sum\nolimits_{j=1}^N\mathsf{a}_{ij} \,(x^i(t)-y^j(t))$};
  \coordinate (x1) at ($(hpre.east)+(\x,0)$) {};
   \node (sum1) at ($(hpre.east)+(1.15*\x,0)$) [sum] {\tiny{$+$}};
  \node (f) at ($(sum1.north)+(0,1.5*\x)$) {$f^i(t)$};
  \node (int)   [block,right=0.5*\x of sum1] {$\int$};
  \node (sum2) at ($(int.east)+(1.9*\x,0)$) [sum] {\tiny{$+$}};
  \node (g) at ($(sum2.north)+(0,1.5*\x)$) {$g^i(t)$};
  
  \node (y) at ($(sum2.east)+(1.8*\x,0)$)  {$y^i(t)$};
  
  \node (one) at ($(int)+(1.05*\x,0)$){};
  \node (two) at ($(one)+(0,-1.5*\x)$){};
  \node (three) at ($(hpre)+(-2.5,-1.5*\x)$){};
  \node (four) at ($(three)+(0,0.7)$){};
  
  \draw [thin] (one.center) -- (two.center);
  \draw [thin] (two.center) -- (three.center);
  \draw [thin] (three.center) -- (four.center);
  \draw [link,thin,photon] (u1) -- ([yshift=1.5mm]hpre.west);
  \draw [link,thin] (four.center) -- ([yshift=-2mm]hpre.west);
  \draw [link,thin] (hpre) -- (sum1);
  \draw [link,thin] (sum1) -- (int);
  \draw [link,thin] (f) -- (sum1);
  \draw [link,thin] (g) -- (sum2);
  \draw [link,thin] (int) to node[midway,above]  {$x^i(t)$}  (sum2);
  \draw [link,thin,photon] (sum2) to (y);
  \vspace{-0.25in}
\end{tikzpicture}
    \caption{{\small Graphical representation of algorithm~\ref{eq::consensus-modified}, where $f^i$ and $g^i$ are the additive obfuscation signals.}}\vspace{-0.1in}
\label{fig::network_island}
 \end{figure} 
}

\begin{thm}[The set of necessary and sufficient conditions on the admissible obfuscation signals] \label{thm::main-col}Consider algorithm~\eqref{eq::consensus-modified} over a strongly connected and weight-balanced digraph with obfuscation signals   $f^i,g^i\in\mathcal{L}^{\infty}_1$, $i\in\VV$. Then, the trajectory $t\mapsto x^i(t)$, of all agents $i\in\VV$ converges to $\frac{1}{N}\sum_{j=1}^N x^j(0)=\avrg{\mathsf{r}}$ as $t\to\infty$ if and only if 
\rm{
\begin{subequations}\label{eq::nec-suf-admin-col}
    \begin{align}
&\!\!\lim_{t\to\infty}\int_{0}^{t}\!\!\sum\nolimits_{k=1}^N(f^k(\tau)\!+\!\dout^k\, g^k(\tau))\,\textup{d}\tau=0,\label{eq::nec-suf-admin-col-a}\\
    &\!\! \lim_{t\to\infty}\int_{0}^{t}\!\!\textup{e}^{-\lL^{+} (t-\tau)}\rR^{\top}(\vect{f}(\tau)+\vectsf{A}\,\vect{g}(\tau))\,\textup{d}\tau=\vect{0},
    \label{eq::nec-suf-admin-col-b}
\end{align}
\end{subequations}
}
where $\lL^{+}$ and $\rR$ are defined in~\eqref{eq::lL+}. \boxend
\end{thm}
The proof of Theorem~\ref{thm::main-col} is given in the appendix. The necessary and sufficient conditions in~\eqref{eq::nec-suf-admin-col} that specify the admissible signals of the agents are highly coupled. If there exists an ultimately secure and trusted authority that oversees the operation, this authority can assign to each agent its admissible private obfuscation signals that collectively satisfy~\eqref{eq::nec-suf-admin-col}. However, in what follows, we consider a scenario where such an authority does not exist, and each agent $i\in\VV$, to increase its privacy protection level, wants to choose its own admissible signals $(f^i,g^i)$ privately without revealing them explicitly to the other agents.

\begin{thm}[Linear algebraic coupling] \label{thm::main}Consider algorithm~\eqref{eq::consensus-modified} over a strongly connected and weight-balanced digraph. Let each agent $i\in\VV$ choose its local obfuscation signals $f^i,g^i\in\mathcal{L}^{\infty}_1$ such that
\begin{align}\label{eq::local_beta-i}
    \lim_{t\to\infty}\int_{0}^{t}\!\!(f^i(\tau)\!+\!\dout^i\, g^i(\tau))\,\textup{d}\tau=\beta^i,
\end{align}
where $\beta^i\in\real$. Then, the necessary and sufficient conditions to satisfy~\eqref{eq::nec-suf-admin-col} are
\rm{
\begin{subequations}\label{eq::nec-suf-admin-sig}
    \begin{align}
&\!\!\sum\nolimits_{k=1}^N\beta^k=0,\label{eq::nec-suf-admin-sig-a}\\
&\!\! \lim_{t\to\infty}\int_{0}^{t} \textup{e}^{-(t-\tau)}{g}^i(\tau)\,\textup{d}\tau=\alpha\in\real,\quad i\in\VV.
    \label{eq::nec-suf-admin-sig-b}
\end{align}
\end{subequations}
}
\boxend
\end{thm}
The proof of Theorem~\ref{thm::main} is given in the appendix. In Theorem~\ref{thm::main}, by enforcing condition~\eqref{eq::local_beta-i} on the admissible signals the coupling between the agents becomes a set of linear algebraic constraints. For a given set of $\{\beta^i\}_{i=1}^N$ and $\alpha$, Theorem~\ref{thm::main} enables the agents to choose their admissible obfuscation signals locally with guaranteed convergence to the exact average consensus; see Remark~\ref{rem::local_fg} below. Choosing signals that satisfy condition~\eqref{eq::local_beta-i} is rather easy. However, condition~\eqref{eq::nec-suf-admin-sig-b} appears to be more complex. 
The result below, whose proof is given in the appendix, identifies three classes of signals that are guaranteed to satisfy condition~\eqref{eq::nec-suf-admin-sig-b}. 
\begin{lem}[Signals that satisfy~\eqref{eq::nec-suf-admin-sig-b} ]\label{lem::admissible-signal}
For a given  $\alpha\!\in\!\real$, let $g=g_1+g_2\in\mathcal{L}^{\infty}_1$ satisfy one of the conditions (a)  $\lim_{t\to\infty}g(t)\!=\!\alpha$ (b) 
$\lim_{t\to\infty}g_1(t)\!=\!\alpha$ and {\rm{$\lim_{t\to\infty}\int_{0}^t\! g_2(\tau)\textup{d}\tau=\bar{g}\!<\!\infty$}}  (c) $\lim_{t\to\infty}g_1(t)\!=\!\alpha$ and {\rm{$\int_{0}^t \sigma(|g_2(\tau)|)\textup{d}\tau<\infty$}} for $t\in\real_{\geq0}$, where $\sigma$ is any class $\mathcal{K}_{\infty}$ function. Then, {\rm{$\lim_{t\to\infty}\int_{0}^t\! \textup{e}^{-(t-\tau)}g(\tau)\textup{d}\tau=\alpha$}}.\boxend
\end{lem}
An interesting theoretical finding that  Lemma~\ref{lem::admissible-signal} reveals is that the admissible obfuscation signals $\{f^j,g^j\}_{j=1}^N$, unlike some of the existing results 
do not necessarily need to be vanishing signals even for $\alpha=0$ and $\beta^i=0$, $i\in\VV$. For example, $g_1(t)=0$ and $g_2(t)=\sin(\phi_0+2\pi(\frac{c}{2}t^2+\omega_0t))$, which is a waveform with linear chirp function~\cite{PF:18} where $\omega_0$ is the starting frequency at time $t=0$, $c\in\real$ is the chirpyness constant, and $\phi_0$ is the initial phase, satisfy condition (b)  of~Lemma~\ref{lem::admissible-signal} with $\alpha=0$. This function is smooth but loses its  uniform continuity as $t\to\infty$. However, when a non-zero $\alpha$ is used the choices for non-vanishing $g$ that satisfy~\eqref{eq::nec-suf-admin-sig-b} are much wider, e.g., according to condition (b) of~Lemma~\ref{lem::admissible-signal} any function that asymptotically converges to $\alpha$ can be used.

\begin{rem}[Locally chosen admissible signals]\label{rem::local_fg}\rm{
If in a network the agents do not know whether others are going to use obfuscation signals or not, then the agents 
use $\alpha=0$ and $\beta^i\!=\!0$, $i\in\VV$ in~\eqref{eq::nec-suf-admin-sig} and \eqref{eq::local_beta-i}.  This is because the only information available to the agents is that their collective choices should satisfy~\eqref{eq::nec-suf-admin-col}. 
Then, in light of Theorem~\ref{thm::main}, 
to ensure~\eqref{eq::nec-suf-admin-col-a} each agent  $i\!\in\!\VV$ chooses its local admissible obfuscation signals according to~\eqref{eq::local_beta-i} with $\beta^i\!=\!0$. Consequently, according to Theorem~\ref{thm::main} again, each agent $i\in\VV$ needs to choose its respective $g^i$ according to~\eqref{eq::nec-suf-admin-sig-b} with $\alpha\!=\!0$. Any other choice of $\{\beta^i\}_{i=1}^N$ and $\alpha$ needs an inter-agent coordination/agreement procedure.
We refer to the admissible signals chosen according to~\eqref{eq::local_beta-i} and~\eqref{eq::nec-suf-admin-sig} as the \emph{locally chosen admissible signals}. 
\boxend}
\end{rem}

In the case of the locally chosen admissible obfuscation signals without inter-agent coordination, since the agents need to satisfy~\eqref{eq::local_beta-i} and~\eqref{eq::nec-suf-admin-sig} with $\alpha\!=\!\beta^i\!=\!0$, $i\in\VV$, these values~will be known to the eavesdroppers. In case that the agents coordinate to choose non-zero values for $\alpha$ and $\{\beta\}_{i=1}^N$ such that~\eqref{eq::local_beta-i} and~\eqref{eq::nec-suf-admin-sig} are satisfied, it is likely that these choices to be known to the eavesdroppers. In our privacy preservation analysis below, we consider various cases of the choices of $\alpha$ and/or $\{\beta\}_{i=1}^N$ being either known or unknown to the eavesdroppers. This way, our study explains the privacy preservation against the most informed eavesdroppers and also explores what kind of guarantees exists against less informed eavesdroppers that do not know all the parameters. The knowledge sets that we consider are defined as follows.  

 \begin{defn}[Knowledge set of an eavesdropper]\label{def::know}\rm{The knowledge set of the internal eavesdropper $1$ and external eavesdropper $\textup{ext}$ is assumed to be one of the cases below,
\begin{itemize}[leftmargin=*]
    \item Case 1:
    \begin{align}\label{eq::knowledge_set-1}
\mathcal{K}^a\!=\,&\left\{\mathcal{Y}^a(\infty),\mathcal{G}(\VV,\EE,\vectsf{A}),\right.\nonumber\\
~&\quad\left.\textup{form of conditions}~\eqref{eq::local_beta-i}~\text{and}~\eqref{eq::nec-suf-admin-sig},\alpha,\{\beta^i\}_{i=1}^N\right\}\!,
\end{align}
\item Case 2:
  \begin{align}\label{eq::knowledge_set-2}
\!\!\!\mathcal{K}^1\!=\,&\left\{\mathcal{Y}^1(\infty),\mathcal{G}(\VV,\EE,\vectsf{A}),\right.\nonumber\\
&\quad\quad\left.\textup{form of conditions}~\eqref{eq::local_beta-i}~\text{and}~\eqref{eq::nec-suf-admin-sig},\alpha\right\}\!,
\end{align}
\begin{align}\label{eq::knowledge_set-4}
\!\!\!\mathcal{K}^{\textup{ext}}\!=\,&\left\{\mathcal{Y}^{\textup{ext}}(\infty),\mathcal{G}(\VV,\EE,\vectsf{A}),\right.\qquad\nonumber\\&\quad~~~\left.\textup{form of conditions}~\eqref{eq::local_beta-i}~\text{and}~\eqref{eq::nec-suf-admin-sig}\right\}\!,
\end{align}
\item Case 3:
\begin{align}\label{eq::knowledge_set-3}
\!\!\!\mathcal{K}^{\textup{ext}}\!=\,&\left\{\mathcal{Y}^{\textup{ext}}(\infty),\mathcal{G}(\VV,\EE,\vectsf{A}),\right.\qquad\nonumber\\&\quad~~~\left.\textup{form of conditions}~\eqref{eq::local_beta-i}~\text{and}~\eqref{eq::nec-suf-admin-sig},\{\beta^i\}_{i=1}^N\right\}\!,
\end{align}
\end{itemize}
where  $a\in\{1,\textup{ext}\}$.}\boxend
\end{defn}
Given internal and external eavesdroppers with knowledge sets belonging to one of the cases in Definition~\ref{def::know}, 
 our study intends to determine: (a) whether the eavesdroppers inside or outside the network can obtain the reference value of the other agents by storing and processing the transmitted messages; (b) more specifically, what \emph{knowledge set} enables an agent inside or outside the network to discover the reference value of the other agents in the network;
 (c) what observers such agents can employ to obtain the reference value of the other agents in the network. 


\section{Privacy preservation evaluation}
In this section, we evaluate the privacy preservation properties of the modified average consensus algorithm~\eqref{eq::consensus-modified} against an internal eavesdropper $1$ and an external eavesdropper whose knowledge sets are either of the two cases given in Definition~\ref{def::know}. From the perspective of an eavesdropper interested in private reference value of another agent $i\in\VV$, the dynamical system to observe is~\eqref{eq::consensus-modified} with  ${x}^i$ as the internal state, ($f^i$, $g^i$, $\{{y}^j\}_{j\in\mathcal{N}_{\textup{out}}^i}$) as the inputs and ${y}^i$ as the measured output. 
When inputs and measured outputs over some finite time interval (resp. infinite time) are known, 
the traditional observability (resp. detectability) tests (see~\cite{RH-AJK:77},\cite{EDS:13}) can determine whether the initial conditions of the system can be identified. However, here the inputs $f^i$ and $g^i:\real_{\geq0}\to\real$  of agent $i\in\VV$ are not available to the eavesdropper. All is known is the conditions~\eqref{eq::local_beta-i} and~\eqref{eq::nec-suf-admin-sig} that specify the obfuscation signals. With regard to inputs $\{{y}^j\}_{j\in\mathcal{N}_{\textup{out}}^i}$ and output $y^i$, an external agent should intercept these signals while the internal eavesdropper $1$ has only access to these inputs if it is an in-neighbor of agent $i$ and all the out-neighbors of agent $i$ (e.g., in Fig.~\ref{fig::k-anom}, agent $1$ is an in-neighbor of agent $2$ and all the out-neighbors of agent~$2$).  

\subsection{Case 1 knowledge set}
Identifying the initial condition of the agents in the presence of unknown additive obfuscation signals may appear to be related to the classical concept of strong observability/detectability in control theory~\cite{MH-RJP:98,MLJH:83}. However, the necessary conditions on the unknown admissible obfuscation signals provide additional information to the eavesdropper. Such information is not being captured by the strong observability/detectability framework, rendering it inadequate for our study. 

It may appear that identifying the initial condition of the agents in the presence of unknown additive obfuscation signals is related to the classical concept of strong observability/detectability in control theory~\cite{MH-RJP:98,MLJH:83}. However, the necessary conditions on the unknown admissible obfuscation signals~\eqref{eq::nec-suf-admin-col} provide additional information to the eavesdropper. Such information is not being captured by the strong observability/detectability framework, rendering it inadequate for our study. 

 \begin{figure}[t]
 \unitlength=0.5in \centering
  \captionsetup[subfloat]{captionskip=1pt}

\begin{tikzpicture}[auto,thick,scale=0.5, every node/.style={scale=0.6}]

\node (null) at (-1,-0.5){};
\node (null) at (5,0){};
             \node (1adv) at (0,0)  [draw, minimum size=10pt,color=white, circle, very thick,dashed] {};
 
\node (1) at (0,0) [draw, minimum size=20pt,color=blue, circle, very thick,fill=red!30] {{$1$}};

\node (2) at (4,1.5) [draw, minimum size=30pt,color=blue, circle,  line width=1.2mm] {{\large${\mathcal{V}}^{\underline{1}}_{k,2}$}};

\node (3) at (8,2) [ draw, minimum size=30pt,color=blue, circle,  line width=1.2mm] {{\large${\mathcal{V}}^{\underline{1}}_{k,3}$}};

\node (4) at (4,-2) [draw, minimum size=30pt,color=blue, circle, line width=1.2mm] {{\large${\mathcal{V}}^{\underline{1}}_{k,4}$}};

\node (5) at (-3.5,-0.8) [ label=below: {{rest of network}},draw, minimum size=45pt,color=gray, circle, line width=1.2mm,dashed] {$\VV\backslash{\mathcal{V}}^{\underline{1}}_k$};

\draw [->,ultra thick] (1) to [out=15+45,in=180+15] (2);
\draw [->,thin] (2) to [out=180+15+45,in=15] (1);

\draw [->,ultra thick] (2) to [out=15+45,in=180-15] (3);
\draw [->,thin] (3) to [out=180+15+45,in=-15] (2);

\draw [->,ultra thick] (1) to [out=-45+15,in=90+45] (4);
\draw [->,thin] (4) to [out=180+15,in=-90+15] (1);

\draw [->,thin] (2) to [out=-90-15,in=90+15] (4);
\draw [->,thin] (4) to [out=90-15,in=-90+15] (2);

\draw [->,thin] (3) to [out=0,in=0] (4);

\draw [->,thin] (3) to [out=90,in=120] (1);

\draw [dashed, gray,ultra thick][->] (1) to [out=180-15,in=90-15] (5);
\draw [dashed,gray, ultra thick][->] (5) to [out=-15,in=270-15] (1);

\end{tikzpicture}
    \caption{{\small The $k^{th}$ induced island of eavesdropper $1$.  The super node $\VV_{k,2}^{\underline{1}}$ in $\GG_{k}^{\underline{1}}$ is the set of the out-neighbors of agent $1$ that  each of them has at least one out-neighbor that is not an out-neighbor of agent $1$. The super node $\VV_{k,4}^{\underline{1}}$ is the set of the out-neighbors of agent $1$ whose out-neighbors are  all also out-neighbors of agent $1$. Finally, the super node $\VV_{k,3}^{\underline{1}}$  is the set of the agents in $\GG_{k}^{\underline{1}}$ that are not an out-neighbor of agent $1$. An arrow from each node $a$ (agent 1 or each super node) to another node $b$ (agent 1 or each super node) indicates that at least one agent in $a$ can obtain information from at least one agent in $b$. 
    The thin connection lines may or may not exist in a network. }\vspace{-0.25in}}
\label{fig::network_island}

 \end{figure} 

Consider the internal eavesdropper, agent 1, when it intends to obtain the initial condition of one of the agents $i \in \VV$. The critical part of the knowledge set of an eavesdropper when it targets an agent is the signals that it has access to.
To study privacy preservation for agent $i \in \VV$, we partition the graph into islands whose nodes are classified into different groups based on their information exchange by the eavesdropper and its out-neighbors, see Fig.~\ref{fig::network_island}.
For that, note that 
removing eavesdropper agent $1$ and its incident edges results in $\bar{n}^1\geq 1$ disjoint subgraphs $\bar{\mathcal{G}}^{\underline{1}}_{k}=(\bar{\mathcal{V}}^{\underline{1}}_{k},\bar{\mathcal{E}}^{\underline{1}}_{k})\subset\mathcal{G}(\VV,\EE)$, $k\in\until{\bar{n}^1}$. Adding agent $1$ in subgraph $\bar{\mathcal{G}}^{\underline{1}}_{k}$ and including its incident edges to this subgraph results in an island graph $\mathcal{G}^{\underline{1}}_{k}=(\mathcal{V}^{\underline{1}}_{k},\mathcal{E}^{\underline{1}}_{k})\subset\mathcal{G}(\VV,\EE)$ where $\mathcal{V}^{\underline{1}}_k=\bar{\mathcal{V}}^{\underline{1}}_{k}\cup\{1\}$ and  $\mathcal{E}^{\underline{1}}_k=\{(l,j)\in\mathcal{E}|\, l \in\mathcal{V}^{\underline{1}}_k,~ j\in\mathcal{V}^{\underline{1}}_k\}$. Every island of agent $1$ is connected to the rest of the digraph $\GG$ only through agent $1$ (see Fig.~\ref{fig::network_island}).
To simplify the notation, with out loss of generality, carry out the subsequent study for agents in island $k=1$, e.g., $\mathcal{G}^{\underline{1}}_1$. Based on how each agent interacts with agent $1$, we divide the agents of island $\mathcal{G}^{\underline{1}}_1$ into three groups as described below (see Fig.~\ref{fig::network_island})
 \begin{itemize}
   {\setlength\itemindent{-8pt} \item
   ${\mathcal{V}}^{\underline{1}}_{1,2}=\big\{i\in\VV_1^{\underline{1}}\,\big|\,i\in\Nout^1,~\Nout^i\not\subset\mathcal{N}^1_{\textup{out}+1}\big\}$, }
 {\setlength\itemindent{-8pt}    \item ${\mathcal{V}}^{\underline{1}}_{1,3}=\big\{i\in\VV_1^{\underline{1}}\,\big|\,i\notin\Nout^1\big\}$.}
   {\setlength\itemindent{-8pt}  \item
   ${\mathcal{V}}^{\underline{1}}_{1,4}=\big\{i\in\VV_1^{\underline{1}}\,\big|\,i\in\Nout^1,~\Nout^i\subseteq\mathcal{N}^1_{\textup{out}+1}\big\}$.}
    \end{itemize}
    ${\mathcal{V}}^{\underline{1}}_{1,4}$ is the set of the agents that agent $1$ has direct access to all their communication signals, while ${\mathcal{V}}^{\underline{1}}_{1,2}$ and ${\mathcal{V}}^{\underline{1}}_{1,3}$ are set of agents that some of inter-agent communication between them is not available to agent $1$.
Without loss of generality, in what follows we assume that the agents in the network are labeled according to the ordered set $(1,{\mathcal{V}}^{\underline{1}}_{1,2},{\mathcal{V}}^{\underline{1}}_{1,3},{\mathcal{V}}^{\underline{1}}_{1,4},\VV\backslash{\mathcal{V}}^{\underline{1}}_1)$. We let the aggregated states and obfuscation signals of the agents in ${\mathcal{V}}^{\underline{1}}_{1,l}$,  $l\in\{2,3,4\}$, be ${\vect{x}}_l=[x^i]_{i\in{\mathcal{V}}^{\underline{1}}_{1,l}}$, ${\vect{g}}_l=[g^i]_{i\in{\mathcal{V}}^{\underline{1}}_{1,l}}$ and ${\vect{f}}_l=[f^i]_{i\in{\mathcal{V}}^{\underline{1}}_{1,l}}$. Similarly, we let the aggregated states and obfuscation signals of the agents in $\VV\backslash{\mathcal{V}}^{\underline{1}}_1$ be ${\vect{x}}_5=[x^i]_{i\in\VV\backslash{\mathcal{V}}^{\underline{1}}_1}$, $\vect{g}_5=[g^i]_{i\in\VV\backslash{\mathcal{V}}^{\underline{1}}_1}$ and $\vect{f}_5=[f^i]_{i\in\VV\backslash{\mathcal{V}}^{\underline{1}}_1}$. We partition $\lL$, $\vectsf{A}$ and $\vectsf{D}^{\textup{out}}$, respectively, to subblock matrices  $\lL_{ij}$'s, $\vectsf{A}_{ij}$'s and $\vectsf{D}^{\textup{out}}_{ij}$'s in a comparable manner to the partitioned aggregated state $(x^1,\vect{x}_2,\vect{x}_3,\vect{x}_4,\vect{x}_5)$(see~\eqref{eq::island_dynamics}). By definition ${\lL}_{ij}=-\vectsf{A}_{ij}$, $i,j\in\{1,\cdots,5\},~i\neq j$. With the right notation at hand, we present the following result which provides the privacy guarantee according to Definition~\ref{df::privacy} for the agents belonging to ${\mathcal{V}}^{\underline{1}}_{1,2}$ and ${\mathcal{V}}^{\underline{1}}_{1,3}$. Note that, because every agent in $\mathcal{G}^{\underline{1}}_1$ is connected to the rest of the agents in digraph $\mathcal{G}$ only through agent $1$, all the out-neighbors and in-neighbors of agent $2$ are necessarily in $\mathcal{G}^{\underline{1}}_1$. The proof of Lemma~\ref{lem::flat_output} is given in the Appendix.

\begin{lem}[A case of indistinguishable admissible initial conditions for an internal eavesdropper]\label{lem::flat_output}
Let agent $1$ be the internal eavesdropper whose knowledge set is as Definition~\ref{eq::knowledge_set-1}. Let $\mathcal{G}^{\underline{1}}_1=(\mathcal{V}^{\underline{1}}_1,\mathcal{E}^{\underline{1}}_1)$ be an island of agent $1$ that satisfies $\mathcal{V}^{\underline{1}}_{1,2}\neq\{\}$. Consider the modified static average consensus algorithm~\eqref{eq::consensus-modified} over a strongly connected and weight-balanced digraph $\mathcal{G}$ where the agents are implementing $\{x^i(0)=\mathsf{r}^i,f^i,g^i\}_{i=1}^N$, with the locally chosen admissible obfuscation signals $(f^i,g^i)$ satisfying \eqref{eq::local_beta-i} and \eqref{eq::nec-suf-admin-sig}. Consider also an alternative execution of~\eqref{eq::consensus-modified} with 
 $\{{x^i}'(0),{f^i}',{g^i}'\}_{i=1}^N$ satisfying  
\begin{align}
    &{x^1}'(0)={x^1}(0), \,\,\vect{x}_4'(0)=\vect{x}_4(0), \,\, \vect{x}_5'(0)=\vect{x}_5(0) \nonumber \\
        &\vect{x}'_2(0)-\vect{x}_2(0)= -\vectsf{A}_{23}\lL_{33}^{-1}(\vect{x}_3'(0)-\vect{x}_3(0)),
        \label{eq::initial_alter}
\end{align}
and
\begin{align}
    &{f^i}'(t)={f^i}(t), \qquad \qquad \qquad i \in \VV \setminus \mathcal{V}^{\underline{1}}_{1,2} \nonumber \\
        &{f^i}'(t)={f^i}(t) \!\!-\!\!\left[\vectsf{A}_{23}\textup{e}^{-{\lL}_{33}t}(\vect{x}_3'(0)\!\!-\!\!\vect{x}_3(0))\right]_{i-1}, \,\,\,  i \in \mathcal{V}^{\underline{1}}_{1,2} \label{eq::initial_alter_f}
\end{align}
and
\begin{align}    
    &{g^i}'(t)={g^i}(t), \qquad \qquad \qquad  i \in \VV \setminus \mathcal{V}^{\underline{1}}_{1,2} \nonumber \\
    &{g^i}'(t)={g^i}(t)\!\!+\!\! \left[\textup{e}^{-{\vectsf{D}}_{22}t}(\vect{x}_2'(0)\!-\!\vect{x}_2(0)),\right]_{i-1},\,\,\,  i \in \mathcal{V}^{\underline{1}}_{1,2}\label{eq::initial_alter_g}
\end{align}
Then,
\begin{align}\label{eq::unobserv-con}
&    y^{i}(t)={y^{i}}'(t),\quad t\in\real_{\geq 0},\quad \quad\quad~ i\in\VV\backslash\VV_{1,3}^{\underline{1}}.
\end{align}
Moreover,
\begin{align}
  &\sum\nolimits_{i=1}^N{x^i}'(0)=\sum\nolimits_{i=1}^N{x^i}(0)=\sum\nolimits_{i=1}^N\mathsf{r}^i,\label{eq::init-alt-sum}\\
&\lim_{t\to\infty}{x^i}'(t)=\frac{1}{N}\sum\nolimits_{i=1}^N\mathsf{r}^i,\quad\quad\quad i\in\VV.\label{eq::init-alt-converge}
\end{align}
\boxend
\end{lem}

Couple of remarks are in order regarding the results of Lemma~\ref{lem::flat_output}. First notice that in proof of Lemma~\ref{lem::flat_output}, we show that each $({f^i}',{g^i}')$, $i\in\VV$ satisfies the locally chosen admissible obfuscation signals conditions~\eqref{eq::local_beta-i} and \eqref{eq::nec-suf-admin-sig} for the same $\alpha$ and $\beta^i$s used to generate $\{f^i,g^i\}_{i=1}^N$. Next notice that due to~\eqref{eq::initial_alter} for any $\gamma\in\real$, there always exists ${x^i}'(0)$ for $i\in (\VV_{1,2}^{\underline{1}}\cup \VV_{1,3}^{\underline{1}})$ that satisfies $\left|{x^i}'(0)-{x^i}(0)\right|>\gamma$, while 
signals received by the eavesdropper as stated in~\eqref{eq::unobserv-con}, are identical for both execution of the algorithm using  $\{x^i(0)=\mathsf{r}^i,f^i,g^i\}_{i=1}^N$ and $\{{x^i}'(0),{f^i}',{g^i}'\}_{i=1}^N$. This means that the privacy all agents in $ (\VV_{1,2}^{\underline{1}}\cup \VV_{1,3}^{\underline{1}})$ is preserved in accordance with Definition~\ref{df::privacy}.

We can develop similar results, as stated in the corollary below, for an external eavesdropper that does not have direct access to the output signal of some of the out-neighbors of agent $i\in\VV$.

\begin{cor}[A case of indistinguishable admissible initial conditions for an external eavesdropper]\label{cor::flat_output_ext}
Let agent $\mathsf{Ext}$ be the internal eavesdropper whose knowledge set is as Definition~\ref{eq::knowledge_set-1} where the eavesdropper has access to $y^l(t), \, l \in \mathcal{O}$ and agent $k$ where the external eavesdropper does not have access to $y^k(t)$, i.e. $k \not\in \mathcal{O}$. Consider the modified static average consensus algorithm~\eqref{eq::consensus-modified} over a strongly connected and weight-balanced digraph $\mathcal{G}$ where the agents are implementing $\{x^i(0)=\mathsf{r}^i,f^i,g^i\}_{i=1}^N$, with the locally chosen admissible obfuscation signals $(f^i,g^i)$ satisfying \eqref{eq::local_beta-i} and \eqref{eq::nec-suf-admin-sig}. Consider also an alternative execution of~\eqref{eq::consensus-modified} with 
 $\{{x^i}'(0),{f^i}',{g^i}'\}_{i=1}^N$ satisfying
\begin{align}
    &{x^i}'(0)=x^i(0) \quad \quad \quad i\in\VV\backslash \mathcal{N}_{\textup{in}}^k\cup\{k\} \nonumber\\
    &{{x}^i}'(0)-{x}^i(0)= -\frac{\mathsf{a}_{ik}}{\textup{d}_{\textup{out}}^{k}}({{x}^k}'(0)-{x}^k(0)) \quad i\in \mathcal{N}_{\textup{in}}^k \label{eq::init_alter_ext}
\end{align}
and
\begin{align}
    &{f^i}'(t)=f^i(t) \quad \quad \quad i\in\VV\backslash \mathcal{N}_{\textup{in}}^k \cup\{k\} \nonumber \\
    &{f^{i}}'(t)=f^i(t)-\mathsf{a}_{ik}\textup{e}^{-{\textup{d}_{\textup{out}}^{k}}t}({{x}^k}'(0)-{x}^k(0)) \quad i\in \mathcal{N}_{\textup{in}}^k \label{eq::initial_alter_f_ext}
\end{align}
and
\begin{align}
    &{g^i}'(t)=g^i(t)\quad \quad \quad i\in\VV\backslash \mathcal{N}_{\textup{in}}^k \cup\{k\} \nonumber \\
    &{g^i}'(t)={g^i}(t)+\textup{e}^{-\dout^i t}({{x}^k}'(0)-{x}^k(0)) \quad i\in \mathcal{N}_{\textup{in}}^k \label{eq::initial_alter_g_ext}
\end{align}
Then
\begin{align}
&  y^{i}(t)={y^{i}}'(t),\quad t\in\real_{\geq 0},\quad \quad\quad~ i\in\VV\backslash\{k\}. \label{eq::flat_output_ext}
\end{align}
Moreover,
\begin{align}
  &\sum\nolimits_{i=1}^N{x^i}'(0)=\sum\nolimits_{i=1}^N{x^i}(0)=\sum\nolimits_{i=1}^N\mathsf{r}^i,\label{eq::init-alt-sum}\\
&\lim_{t\to\infty}{x^i}'(t)=\frac{1}{N}\sum\nolimits_{i=1}^N\mathsf{r}^i,\quad\quad\quad i\in\VV.\label{eq::init-alt-converge-ext}
\end{align}
\boxend
\end{cor}
\begin{proof}
Proof of Corollary~\ref{cor::flat_output_ext}, is straight forward from proof of Lemma~\ref{lem::flat_output}. The proof is trivially concluded from equations~\eqref{eq::initial_alter_g},\eqref{eq::initial_alter_f}, and \eqref{eq::initial_alter} through singling out an agent in $\mathcal{V}^{\underline{1}}_{1,3}$ and finding all of its in-neighbors in $\mathcal{V}^{\underline{1}}_{1,2}$.
\end{proof}
Corollary~\ref{cor::flat_output_ext} shows $({f^i}',{g^i}')$, $i\in\VV$ satisfies the locally chosen admissible obfuscation signals conditions~\eqref{eq::local_beta-i} and \eqref{eq::nec-suf-admin-sig} for the same $\alpha$ and $\beta^i$s used to generate $\{f^i,g^i\}_{i=1}^N$. Next notice that due to~\eqref{eq::init_alter_ext}, for any $\gamma \in \realpositive$, there always exist ${x^j}'(0), \,\, j \in \mathcal{N}_{\textup{in}}^k$ and ${x^k}'(0)$ that satisfies $\left|{x^j}'(0)-{x^j}(0)\right|>\gamma$ and $\left|{x^k}'(0)-{x^k}(0)\right|>\gamma$, while the signal transmitted by the agents in $\VV \backslash \{k\}$ as stated in~\eqref{eq::flat_output_ext} are identical for both execution of the algorithm using  $\{x^i(0)=\mathsf{r}^i,f^i,g^i\}_{i=1}^N$ and $\{{x^i}'(0),{f^i}',{g^i}'\}_{i=1}^N$. Moreover, since $\mathcal{O}\subset\VV\backslash\{k\}$) leads to the fact that the privacy of the agents $\mathcal{N}_{\textup{in}}^k \cup \{k\}$ is preserved in accordance with Definition~\ref{df::privacy}.

Through Lemma~\ref{lem::flat_output}, we have established that the privacy of agents when the eavesdropper, either internal or external, does not have access to at least one signal that is transmitted in to the agent is preserved. The next results show that such guarantee does not hold for agents whose incoming and outgoing signals are in the knowledge set of the eavesdropper.

\begin{lem}[Observer design for eavesdroppers with the knowledge set of Case 1]\label{lem::easy-prey-external-ad}
Consider the modified static average consensus algorithm~\eqref{eq::consensus-modified} with a set of locally chosen admissible obfuscation signals $\{f^j,g^j\}_{j=1}^N$ over a strongly connected and weight-balanced digraph $\mathcal{G}$. Let the knowledge set of the eavesdroppers be~\eqref{eq::knowledge_set-1}. An internal eavesdropper agent $1$ and external eavesdropper $\textup{ext}$
 that has access to the output signals of agent $i\in\VV$ and all its out-neighbors, can employ respectively observer
\begin{subequations}\label{eq::out-observer}
\begin{align}
    \dot{\psi}&=\sum\nolimits_{j=1}^N\mathsf{a}_{ij}(y^i-y^j),~~~~~ \psi(0)=-\beta^i,\label{eq::out-observer-state}\\
   \nu^1(t)&=\psi(t)+x^1(t),
   \label{eq::out-observer-out}
\end{align}
\end{subequations}
and observer 
\begin{subequations}\label{eq::out-observer-external}
\begin{align}
    \dot{\zeta}&=\sum\nolimits_{j=1}^N\mathsf{a}_{ij}(y^i-y^j), \quad \zeta(0)=-\beta^i-\alpha,\label{eq::out-observer-state-ext1}\\
    \dot{\eta}&=-{\eta}+y^i,\,\,\quad~\quad\qquad\quad\eta(0)\in\real,\label{eq::out-observer-state-ext2}\\
    \nu^{\textup{ext}}(t)&=\zeta(t)+\eta(t),
   \label{eq::out-observer-ext-out}
\end{align}
\end{subequations}
 to asymptotically obtain $\mathsf{r}^i$, $i\in\VV$, i.e., $\nu^a\to\mathsf{r}^i$, $a\in\{\textup{ext},1\}$ as $t\to\infty$. Moreover, at any time $t\in\real_{\geq0}$, the estimation error of the observers respectively satisfies
\rm{
\rm{
\begin{align}\label{eq::error_out_neig_known}
\nu^1(t)\!-\!\mathsf{r}^i\!=\!
x^1(t)\!-\!x^i(t)+\!\!\int_{0}^t\!\!(f^i(\tau)\!+\!\dout^i\,g^i(\tau))\textup{d}\tau\,\!-\!\!\beta^i.
\end{align}
}
and
\begin{subequations}
\begin{align}
&\nu^{\textup{ext}}(t)\!-\!\mathsf{r}^i\!=\eta(t)\!-
\!x^i(t)\!+\!\!\int_{0}^t\!\!\!(f^i(\tau)\!+\!\dout^i\,g^i(\tau))\textup{d}\tau\!-\!\beta^i\!\!-\!\alpha,\label{eq::error_out_neig_known_ext}\\
  &  \eta(t)=\textup{e}^{-t}\eta_0 \!+ \!\!\!\int_0^{t}\!\! \textup{e}^{-(t-\tau)}x^i(\tau)\textup{d}\tau\!+\!\!\! \int_0^{t} \!\!\textup{e}^{-(t-\tau)}g^i(\tau)\textup{d}\tau.\label{eq::external_eta}
\end{align}
\end{subequations}

}
\end{lem}
\begin{proof}
For an internal eavesdropper, given~\eqref{eq::consensus-modified} and~\eqref{eq::out-observer} we can write
\begin{align*}
\dot{\psi}+\dot{x}^i=f^i+\dout^i\,g^i
\end{align*}
which, because of $x^i(0)={\mathsf{r}}^i$ and $\zeta(0)=-\beta^i$, gives 
\begin{align*}
\psi(t)=-x^i(t)+{\mathsf{r}}^i+\!\int_{0}^t\!\!(f^i(\tau)+\dout^i\,g^i(\tau))\textup{d}\tau-\beta^i,~ t\in\real_{\geq0}.
\end{align*}
Then, using~\eqref{eq::out-observer-out} and~\eqref{eq::consensus-modified-y} we obtain~\eqref{eq::error_out_neig_known} as the estimation error.
Subsequently, because of~\eqref{eq::local_beta-i} and since~$\lim_{t\to\infty}(x^1(t)-x^i(t))=0$, from~\eqref{eq::error_out_neig_known} we obtain $
\lim_{t\to\infty}\nu(t)=\mathsf{r}^i$.

For an external eavesdropper, given~\eqref{eq::consensus-modified} and \eqref{eq::out-observer-state-ext1}, we can write
\begin{align*}
\dot{\zeta}+\dot{x}^i=f^i+\dout^i\,g^i,
\end{align*}
which given $x^i(0)\!=\!{\mathsf{r}}^i$ and $\zeta(0)\!=\!-\beta^i-\alpha$, for $t\in\real_{\geq0}$ gives 
\begin{align}\label{eq::eq::external_zeta_x}
\zeta(t)\!=\!-x^i(t)\!+\!{\mathsf{r}}^i\!+\!\!\int_{0}^t\!\!(f^i(\tau)\!+\!\dout^i\,g^i(\tau))\textup{d}\tau\,-\beta^i-\alpha.
\end{align}
On the other hand, using~\eqref{eq::consensus-modified-y}, $t\mapsto \eta(t)$ is obtained from~\eqref{eq::external_eta}. Then, tracking error~\eqref{eq::error_out_neig_known_ext} is readily deduced from~\eqref{eq::out-observer-ext-out} and~\eqref{eq::eq::external_zeta_x}. Next, given~\eqref{eq::local_beta-i} and~\eqref{eq::nec-suf-admin-sig-b} and also $\lim_{t\to \infty}\textup{e}^{-t}\eta_0=0$, we obtain $ \lim_{t\to\infty}\nu(t)={\mathsf{r}}^i+\lim_{t\to\infty}(-x^i(t)+ \!\!\int_0^{t}\! \textup{e}^{-(t-\tau)}x^i(\tau)\textup{d}\tau)$.

Subsequently, since $\lim_{t\to\infty}x^i(t)=\avrg{\mathsf{r}}$, we can conclude our proof by invoking Lemma~\ref{lem::conv_vanishing} that guarantees  $\lim_{t\to\infty}\int_0^{t}\! \textup{e}^{-(t-\tau)}x^i(\tau)\textup{d}\tau=\lim_{t\to\infty}x^i(t)=\avrg{\mathsf{r}}$. 
\end{proof}
To construct observer~\eqref{eq::out-observer}, the internal eavesdropper used its local state. To compensate for the lack of internal state information, the external eavesdropper is forced to employ a higher-order observer~\eqref{eq::out-observer-external} and invoke  condition~\eqref{eq::nec-suf-admin-sig-b}, which the internal eavesdropper does not need. Thus, an external eavesdropper incurs a higher computational cost.

\begin{figure}[t]
    \centering
     \begin{tikzpicture}[auto,thick,scale=0.67, every node/.style={scale=0.67}]
\tikzset{edge/.style = {->,> = latex'}}
el/.style = {inner sep=2pt, align=left, sloped},
\node (null) at (-1,-0.5){};
\node (null) at (5,0){};
             \node (1adv) at (0,0)  [draw, minimum size=10pt,color=white, circle, very thick,dashed] {};
 
\node (1) at (0,0) [ label=above: {$\mathsf{r}^1\!\!$},draw, minimum size=15pt,color=blue, circle, very thick,,fill=red!30] {{\small \textbf{1}}};

\node (2) at (-2.5,-2.0) [ label=below: {$\mathsf{r}^2\!\!$},draw, minimum size=15pt,color=blue, circle, very thick] {{\small 2}};
 
\node (3) at (-1.5,-0.5) [ label=left: {$\mathsf{r}^3\!\!$},draw, minimum size=15pt,color=blue, circle, very thick] {{\small 3}};

\node (4) at (-2.4,1.5) [ label=left: {$\mathsf{r}^4\!\!$},draw, minimum size=15pt,color=blue, circle, very thick] {{\small 4}};

\node (5) at (1.5,-1.5) [ label=left: {$\mathsf{r}^{5}\!\!$},draw, minimum size=15pt,color=blue, circle, very thick] {{\small 5}};

\node (6) at (1.5,1.5) [ label=below: {$\mathsf{r}^6\!\!$},draw, minimum size=15pt,color=blue, circle, very thick] {{\small 6}};

\node (7) at (3,2) [ label=above: {$\mathsf{r}^7\!\!$},draw, minimum size=15pt,color=blue, circle, very thick] {{\small 7}};

\node (8) at (3,0.5) [ label=left: {$\mathsf{r}^8\!\!$},draw, minimum size=15pt,color=blue, circle, very thick] {{\small 8}};

\draw[edge]  (1)to node[right]  {$1$} (2);
\draw[edge]  (1)to node[above]  {$1$} (3) ;
\draw[edge]  (2)to node[left]  {$2$} (3) ;
\draw[edge]  (3)to node[right]  {$3$} (4) ;
\draw[edge]  (4)to node[left]  {$1$} (2) ;
\draw[edge]  (4)to node[above]  {$2$} (1) ;

\draw[latex'-latex']  (1)-- node[above]  {$3$} (5) ;

\draw[latex'-latex']  (1)-- node[above]  {$3$} (6) ;
\draw[edge]  (6)to node[above]  {$1$} (7) ;
\draw[edge]  (7)to node[right]  {$1$} (8) ;
\draw[edge]  (8)to node[above]  {$1$} (6) ;

\draw [thin,dotted] plot [smooth cycle] coordinates {(1,0) (0,-0.5) (-0.5,0) (0,0.5) (0.75,1.5) (1.5,2.50) (3.75,2.75) (3.75,0.00)} ;
\coordinate (A) at (3.75,2.75);
\node at (A) [above=1mm, right=1mm of A] {$\mathcal{G}^{\underline{1}}_{3}$};

\draw [thin,dash dot] plot [smooth cycle] coordinates {(0,-0.5) (-0.5,0) (0,0.5) (0.5,0) (1,-0.5) (2.5,-0.5) (2.5,-2.5) (0.5,-2.5) (0.5,-1)};
\coordinate (B) at (2.5,-2.5);
\node at (B) [below=1mm, right=1mm of B] {$\mathcal{G}^{\underline{1}}_{2}$};

\draw [thin,dashed] plot [smooth cycle] coordinates { (0,-0.75) (0.5,0) (0,0.5) (-0.75,1) (-1.25,2)  (-3,2) (-3,-2.5) (-0.5,-2.5)};
\coordinate (C) at (-3,-2.5);
\node at (C) [below=1mm, left=1mm of C] {$\mathcal{G}^{\underline{1}}_{1}$};
\end{tikzpicture}
    \caption{{\small A strongly connected and weight-balanced digraph $\mathcal{G}$ in which node $1$ is an articulation point of the undirected representation of $\mathcal{G}$.  $\mathcal{G}^{\underline{1}}_{1}$, $\mathcal{G}^{\underline{1}}_{2}$ and $\mathcal{G}^{\underline{1}}_{3}$ are the islands of agent~$1$.}
    }
    \label{fig::k-anom}
\end{figure}

When an eavesdropper does not have direct access to all the signals in $\{y^j(t)\}_{j\in\mathcal{N}_{\textup{out}+i}^i}$, a rational strategy appears to be that the eavesdropper estimates the signals it does not have access to. If those agents also have out-neighbors that their output signals are not available to the eavesdropper,  then the eavesdropper should estimate the state of those agents as well, until the only inputs to the dynamics that it observes are the additive admissible obfuscation signals. For example, in Fig.~\ref{fig::k-anom}, to obtain the reference value of agent $6$, agent $1$ compensates for the lack of direct access to $y^7(t)$, which enter the dynamics of agent $6$, by estimating the state of all the agents in subgraph $\mathcal{G}^{\underline{1}}_{3}$. Our results below however show that this strategy is not effective. In fact, we show that an eavesdropper (internal or external) is able to uniquely identify the reference value of an agent $i\in\VV$ if and only if it has direct access to $\{y^j(t)\}_{j\in\mathcal{N}_{\textup{out}+i}^i}$ for all $t\in\real_{\geq0}$. 

Building on our results so far, we are now ready to state the necessary and sufficient condition under which an eavesdropper  with knowledge set~\eqref{eq::knowledge_set-1} can discover the reference value of an agent $i\in\VV$.
\begin{thm}[Privacy preservation using the modified average consensus algorithm~\eqref{eq::consensus-modified} when the knowledge set of the eavesdroppers is given by Case 1 in Definition~\ref{def::know}]\label{thm::main_alpha_beta_known}
Consider the modified static average consensus algorithm~\eqref{eq::consensus-modified} with a set of locally chosen admissible obfuscation signals $\{f^i,g^i\}_{i=1}^N$ over a strongly connected and weight-balanced digraph $\mathcal{G}$.
Let the knowledge set of the internal eavesdropper $1$ and external agent $\textup{ext}$ be~\eqref{eq::knowledge_set-1}. Then, (a) agent $1$ can reconstruct the exact initial value of agent $i\in\VV\backslash\{1\}$ if and only if $i\in\Nout^1$ and $\Nout^i\subseteq \mathcal{N}_{\textup{out}+1}^1$; (b) the external agent $\textup{ext}$ can reconstruct the exact initial value of agent $i\in\VV$ if and only if 
$\{\{y^j(\tau)\}_{j\in\mathcal{N}_{\textup{out}+i}^i}\}_{\tau=0}^\infty\subseteq \mathcal{Y}^{\textup{ext}}(\infty)$.
\end{thm}
\begin{proof}
Proof of statement (a): If $i\in\Nout^1$ and $\Nout^i\subseteq \mathcal{N}_{\textup{out}+1}^1$, Lemma~\eqref{lem::easy-prey-external-ad} guarantees that agent $1$ can employ an observer to obtain the reference value of agent $i$.  Next, we show that if $i\not\in\Nout^1$ or $\Nout^i\not\subset \mathcal{N}_{\textup{out}+1}^1$, then agent $1$ cannot uniquely identify the reference value $\mathsf{r}^i$ of agent $i$. Suppose agent $i\in\VV\backslash\{1\}$ satisfies $i\not\in\Nout^1$ (resp. $i\in \mathcal{N}_{\textup{out}}^1$ and $\Nout^i\not\subset \mathcal{N}_{\textup{out}+1}^1$). Without loss of generality let $\VV_{1}^{\underline{1}}$ be the island of agent $1$ that contains this agent $i$. Consequently, $i\in\VV_{1,3}^{\underline{1}}$ (resp. $i\in\VV_{1,2}^{\underline{1}}$). Then, by virtue of Lemma~\ref{lem::flat_output}, we  know that there exists infinite number of alternative admissible initial conditions and corresponding admissible obfuscation signals for any agents in $\VV_{1,3}^{\underline{1}}\cup\VV_{1,2}^{\underline{1}}$  for  which the time histories of each signal transmitted to agent $1$ are identical. Therefore, agent $1$ cannot uniquely identify the initial condition of any agents in $ \VV_{1,3}^{\underline{1}}\cup\VV_{1,2}^{\underline{1}}$. In light of Lemma~\ref{lem::easy-prey-external-ad} and Corollary~\ref{cor::flat_output_ext}, the proof of statement (b) is similar to that of statement (a) and is omitted for brevity.  
\end{proof} 

\begin{rem}[Privacy preserving graph topologies]{\rm 
 There are several classes of undirected graphs for which any two agents on the graph have an exclusive neighbor with respect to the other. Thus, by Theorem~\ref{thm::main_alpha_beta_known} privacy of all the agents is preserved from any internal eavesdropper when they implement algorithm~\eqref{eq::consensus-modified}. Examples include cyclic bipartite undirected graphs, 4-regular ring lattice undirected graphs with $N>5$, planar stacked prism graphs, directed ring graphs, and any biconnected undirected graph that does not contain a cycle with $3$ edges  (see~\cite{CRR-RJW:05} for the formal definition of these graph topologies). Some examples of these privacy-preserving topologies are shown in  Fig.~\ref{Fig::private_graphs}.
Theorem~\ref{thm::main_alpha_beta_known} also presents an opportunity to make agents private with respect to a particular or all the other agents by rewiring the graph so that the conditions of the theorem are satisfied. The idea of rewiring the graph to induce privacy preservation has been explored in the literature~\cite{DIR-RAF-KML:19,YX-Zl:20,SZ-TOT-MAD:20,ILDR-RAF-KML:20}. However, in practice, rewiring may be infeasible or costly.
}
\boxend\end{rem}

\begin{figure}[t!]
  \subfloat[A cyclic 
  bipartite undirected connected 
   graph.]{
    \begin{tikzpicture}[auto,thick,scale=0.67, every node/.style={scale=0.67}]
\node (null) at (-1.5,-0.5){};
\node (null2) at (4,0){};
\node (1) at (0,0) [draw, minimum size=15pt,color=blue, circle, very thick,fill=blue!10]{{\small \textbf{1}}};

\node (2) at (0,2) [draw, minimum size=15pt,color=blue, circle, very thick,fill=blue!10]{{\small \textbf{2}}};

\node (3) at (0,4) [draw, minimum size=15pt,color=blue, circle, very thick,fill=blue!10]{{\small \textbf{3}}};

\node (4) at (2,0) [draw, minimum size=15pt,color=blue, circle, very thick]{{\small \textbf{4}}};

\node (5) at (2,2) [draw, minimum size=15pt,color=blue, circle, very thick]{{\small \textbf{5}}};

\node (6) at (2,4) [draw, minimum size=15pt,color=blue, circle, very thick]{{\small \textbf{6}}};
 
\draw (1)--(4) ;
\draw (1)--(5) ;
\draw (1)--(6) ;

\draw (2)--(4) ;
\draw (2)--(5) ;
\draw (2)--(6) ;

\draw (3)--(4) ;
\draw (3)--(5) ;
\draw (3)--(6) ;
\end{tikzpicture}
  }\qquad
  \subfloat[A 4-regular ring lattice 
  undirected  connected graph 
  on 9 vertices.]
  {
     \begin{tikzpicture}[auto,thick,scale=0.67, every node/.style={scale=0.67}]
\node (null) at (-1,-0.5){};
\node (1) at ({2*sin(0*360/9)},{2*cos(0*360/9)}) [draw, minimum size=15pt,color=blue, circle, very thick] {{\small \textbf{1}}};
  
\node (2) at ({2*sin(1*360/9)},{2*cos(1*360/9)}) [draw, minimum size=15pt,color=blue, circle, very thick] {{\small 2}};

 \node (3) at ({2*sin(2*360/9)},{2*cos(2*360/9)}) [draw, minimum size=15pt,color=blue, circle, very thick] {{\small 3}};

\node (4) at ({2*sin(3*360/9)},{2*cos(3*360/9)}) [draw, minimum size=15pt,color=blue, circle, very thick] {{\small 4}};

\node (5) at ({2*sin(4*360/9)},{2*cos(4*360/9)}) [draw, minimum size=15pt,color=blue, circle, very thick] {{\small 5}};

\node (6) at ({2*sin(5*360/9)},{2*cos(5*360/9)}) [draw, minimum size=15pt,color=blue, circle, very thick] {{\small 6}};

\node (7) at ({2*sin(6*360/9)},{2*cos(6*360/9)}) [draw, minimum size=15pt,color=blue, circle, very thick] {{\small 7}};

\node (8) at ({2*sin(7*360/9)},{2*cos(7*360/9)}) [draw, minimum size=15pt,color=blue, circle, very thick] {{\small 8}};

\node (9) at ({2*sin(8*360/9)},{2*cos(8*360/9)}) [draw, minimum size=15pt,color=blue, circle, very thick] {{\small 9}};

\draw (1)--(2) ;
\draw (1)--(9) ;
\draw (1)--(8) ;
\draw (1)--(3) ;

\draw (2)--(3) ;
\draw (2)--(1) ;
\draw (2)--(4) ;
\draw (2)--(9) ;

\draw (3)--(4) ;
\draw (3)--(2) ;
\draw (3)--(5) ;
\draw (3)--(1) ;

\draw (4)--(5) ;
\draw (4)--(3) ;
\draw (4)--(6) ;
\draw (4)--(2) ;

\draw (5)--(6) ;
\draw (5)--(4) ;
\draw (5)--(7) ;
\draw (5)--(3) ;

\draw (6)--(7) ;
\draw (6)--(5) ;
\draw (6)--(8) ;
\draw (6)--(4) ;

\draw (7)--(8) ;
\draw (7)--(6) ;
\draw (7)--(9) ;
\draw (7)--(5) ;

\draw (8)--(9) ;
\end{tikzpicture}
  }\\ \vspace{-0.08in}
 \subfloat[A triangular stacked prism graph.]{
  \begin{tikzpicture}[auto,thick,scale=0.67, every node/.style={scale=0.67}]
\node (null) at (-1.5,-0.5){};
         
\node (1) at ({2.5*sin(0*360/3)},{2.5*cos(0*360/3)}) [draw, minimum size=15pt,color=blue, circle, very thick] {{\small \textbf{1}}};
  
\node (2) at ({2.5*sin(1*360/3)},{2.5*cos(1*360/3)}) [draw, minimum size=15pt,color=blue, circle, very thick] {{\small 2}};

 \node (3) at ({2.5*sin(2*360/3)},{2.5*cos(2*360/3)}) [draw, minimum size=15pt,color=blue, circle, very thick] {{\small 3}};

\node (4) at ({1.5*sin(0*360/3)},{1.5*cos(0*360/3)}) [draw, minimum size=15pt,color=blue, circle, very thick] {{\small 4}};

\node (5) at ({1.5*sin(1*360/3)},{1.5*cos(1*360/3)}) [draw, minimum size=15pt,color=blue, circle, very thick] {{\small 5}};

\node (6) at ({1.5*sin(2*360/3)},{1.5*cos(2*360/3)}) [draw, minimum size=15pt,color=blue, circle, very thick] {{\small 6}};

\node (7) at ({0.5*sin(0*360/3)},{0.5*cos(0*360/3)}) [draw, minimum size=15pt,color=blue, circle, very thick] {{\small 7}};

\node (8) at ({0.5*sin(1*360/3)},{0.5*cos(1*360/3)}) [draw, minimum size=15pt,color=blue, circle, very thick] {{\small 8}};

\node (9) at ({0.5*sin(2*360/3)},{0.5*cos(2*360/3)}) [ draw, minimum size=15pt,color=blue, circle, very thick] {{\small 9}};

\draw (1)--(2) ;
\draw (2)--(3) ;
\draw (3)--(1) ;

\draw (4)--(5) ;
\draw (5)--(6) ;
\draw (6)--(4) ;

\draw (7)--(8) ;
\draw (8)--(9) ;
\draw (9)--(7) ;

\draw (1)--(4) ;
\draw (4)--(7) ;

\draw (2)--(5) ;
\draw (5)--(8) ;

\draw (3)--(6) ;
\draw (6)--(9) ;
\end{tikzpicture}
  }\qquad
  \subfloat[A lattice graph with 16 vertices (a biconnected graph that contains no cycle with 3 edges).]
  {
      \begin{tikzpicture}[auto,thick,scale=0.67, every node/.style={scale=0.67}]
\node (null) at (-3,-0.5){};
\node (null) at (2.2,0){};

\node (1) at (-2,1) [draw, minimum size=15pt,color=blue, circle, very thick,] {{\small \textbf{1}}};
  
\node (2) at (-2,0) [draw, minimum size=15pt,color=blue, circle, very thick] {{\small 2}};

 \node (3) at (-2,-1) [draw, minimum size=15pt,color=blue, circle, very thick] {{\small 3}};
 
 \node (4) at (-2,-2) [draw, minimum size=15pt,color=blue, circle, very thick] {{\small 4}};

\node (5) at (-1,1) [draw, minimum size=15pt,color=blue, circle, very thick] {{\small 5}};

\node (6) at (-1,0) [draw, minimum size=15pt,color=blue, circle, very thick] {{\small 6}};

\node (7) at (-1,-1) [draw, minimum size=15pt,color=blue, circle, very thick] {{\small 7}};

\node (8) at (-1,-2) [draw, minimum size=15pt,color=blue, circle, very thick] {{\small 8}};

\node (9) at (0,1) [draw, minimum size=15pt,color=blue, circle, very thick] {{\small 9}};

\node (10) at (0,0) [draw, minimum size=15pt,color=blue, circle, very thick,inner sep=2.2pt] {{\small 10}};

\node (11) at (0,-1) [ draw, minimum size=15pt,color=blue, circle, very thick,inner sep=2.2pt] {{\small 11}};

\node (12) at (0,-2) [ draw, minimum size=15pt,color=blue, circle, very thick,inner sep=2.2pt] {{\small 12}};

\node (13) at (1,1) [draw, minimum size=15pt,color=blue, circle, very thick,inner sep=2.2pt] {{\small 13}};

\node (14) at (1,0) [draw, minimum size=15pt,color=blue, circle, very thick,inner sep=2.2pt] {{\small 14}};

\node (15) at (1,-1) [ draw, minimum size=15pt,color=blue, circle, very thick,inner sep=2.2pt] {{\small 15}};

\node (16) at (1,-2) [ draw, minimum size=15pt,color=blue, circle, very thick,inner sep=2.2pt] {{\small 16}};

\draw (1)--(2) ;
\draw (1)--(5) ;
\draw (2)--(3) ;
\draw (2)--(6) ;
\draw (3)--(4) ;
\draw (3)--(7) ;
\draw (4)--(8) ;
\draw (5)--(6) ;
\draw (5)--(9) ;
\draw (6)--(7) ;
\draw (6)--(10) ;
\draw (7)--(8) ;
\draw (7)--(11) ;
\draw (8)--(12) ;
\draw (9)--(10) ;
\draw (9)--(13) ;
\draw (10)--(11) ;
\draw (10)--(14) ;
\draw (11)--(12) ;
\draw (11)--(15) ;
\draw (12)--(16) ;
\draw (13)--(14) ;
\draw (14)--(15) ;
\draw (15)--(16) ;
\end{tikzpicture}
  }
     \caption{{\small
  Examples of privacy-preserving graph topologies.}
 }\label{Fig::private_graphs}
\end{figure}

Next, we show that even though agent $1$ cannot obtain the initial condition of the individual agents in $\VV_{k,2}^{\underline{1}}\neq\{\}$ and $\VV_{k,3}^{\underline{1}}$, $k\in\{1,\cdots\bar{n}^1\}$, it can obtain the average of the initial conditions of those agents. 
 Without loss of generality, we demonstrate our results for $k=1$. 
 
\begin{prop}[Island anonymity]\label{eq::k-anom}
Consider the dynamic consensus algorithm~\eqref{eq::consensus-modified} over a strongly connected and weight-balanced digraph $\mathcal{G}$ in which  
$\mathcal{V}^{\underline{1}}_{1,2}\neq\{\}$.  Let $n_{2,3}=|\VV^1_{1,2}\cup\VV^1_{1,3}|$ and $\dout^{1,1}=\sum\limits_{j\in(\VV^1_{1,2}\cup\VV^1_{1,4})}\!\!\!\!\mathsf{a}_{1j}$ be the out-degree of agent $1$ in subgraph ${\mathcal{G}}^{\underline{1}}_1$. Then, the eavesdropper $1$ with the knowledge set~\eqref{eq::knowledge_set-1} can employ the observer
\rm{\begin{align*}
    &\dot{\zeta_i}=\sum\nolimits_{j=1}^N\mathsf{a}_{ij}(y^i-y^j),~~~~~~~ \zeta_i(0)=-\beta^i,\quad i\in\VV_{1,4}^{\underline{1}},\\
    &\dot{\eta}=\,-\!\!\!\!\!\!\!\sum_{j\in({\mathcal{V}}^{\underline{1}}_{1,2}\cup{\mathcal{V}}^{\underline{1}}_{1,4})}\!\!\!\!\!\!\!\!a_{1j}(y^1-y^j),\quad~~~ \eta(0)=-\!\!\sum\nolimits_{j\in\mathcal{V}_1^{\underline{1}}\backslash\{1\}}\!\! \beta^i,\\
    &\mu(t)=\frac{\eta(t)-\sum\nolimits_{i\in\VV_{1,4}^{\underline{1}}}\zeta_i}{n_{2,3}}+x^1(t).
\end{align*}}
to have $\lim_{t\to\infty}\mu(t)=\frac{1}{n_{2,3}}\sum\limits_{j\in(\VV^{\underline{1}}_{1,2}\cup\VV^{\underline{1}}_{1,3})}\!\!\!\mathsf{r}^j$.
\end{prop}
\begin{proof}
Consider the aggregate dynamics of $\eta$ and $\vect{x}_i$, $i\in\{2,3,4\}$, which reads as
\begin{align*}
    \begin{bmatrix}
    \dot{\eta}\\
    \dvect{x}_2\\
    \dvect{x}_3\\
    \dvect{x}_4
    \end{bmatrix}=&-\underbrace{\begin{bmatrix}\dout^{1,1} & -\vectsf{A}_{12}&\vect{0}&-\vectsf{A}_{14}\\
    -\vectsf{A}_{21}&\Dout_{22}&-\vectsf{A}_{23}&-\vectsf{A}_{24}\\
    -\vectsf{A}_{31}&-\vectsf{A}_{32}&\Dout_{33}&\vect{0}\\
    -\vectsf{A}_{41}&-\vectsf{A}_{42}&\vect{0}&\Dout_{44}
    \end{bmatrix}}_{\lL_{1}^{\underline{1}}}\begin{bmatrix}
    y^1\\
    \vect{y}_2\\\vect{y}_3\\\vect{y}_4
    \end{bmatrix}+\\
    &\begin{bmatrix}
    0\\
    \vect{f}_2+\Dout_{22}\vect{g}_2\\
    \vect{f}_3+\Dout_{33}\vect{g}_3\\
    \vect{f}_4+\Dout_{44}\vect{g}_4
    \end{bmatrix}.
\end{align*}
Notice that ${\lL_{1}^{\underline{1}}}$ is the Laplacian matrix of graph $\mathcal{G}^{\underline{1}}_1$. By Virtue of Lemma~\ref{lem::subgraph-balanced} in the appendix we know that $\mathcal{G}^{\underline{1}}_1$ is a strongly connected and weight-balanced digraph. Consequently, left multiplying both sides of equation above with $\vect{1}_{|\VV_{1}^{\underline{1}}|}^\top$ gives
\begin{align*}
\dot{\eta}\!+\sum_{j\in\mathcal{V}_1^{\underline{1}}\backslash\{1\}}\!\!\!\!\! x^i=
&\!\!\!\!\sum\limits_{j\in\mathcal{V}_1^{\underline{1}}\backslash\{1\}}\!\!\!\!\! (f^j(t)+\dout^j\, g^j(t)).
\end{align*}
Thereby, given $\eta(0)=-\!\!\!\!\!\sum\limits_{j\in\mathcal{V}_1^{\underline{1}}\backslash\{1\}}\!\!\!\!\! \beta^i$ and $x^i(0)=\mathsf{r}^i$, we obtain
\begin{align*}
\eta(t)=\!\!\!\!\!\!\!\sum_{j\in\mathcal{V}_1^{\underline{1}}\backslash\{1\}}\!\!\!\!\!\! \mathsf{r}^j~-\!\!\!\!\!\sum_{j\in\mathcal{V}_1^{\underline{1}}\backslash\{1\}}\!\!\!\!\!\! x^j(t)&+\!\!\!\!\!\!\!\sum\limits_{j\in\mathcal{V}_1^{\underline{1}}\backslash\{1\}}\!\!\!\!\! \int_{0}^t\!(f^j(\tau)+\dout^j\, g^j(\tau))\textup{d}\tau\\
&-\!\!\!\sum\nolimits_{j\in\mathcal{V}_1^{\underline{1}}\backslash\{1\}}\!\! \beta^i.
\end{align*}
On the other hand, following the proof of Lemma~\ref{lem::easy-prey-external-ad}, we can conclude that 
\begin{align*}
    \sum\limits_{i\in\VV_{1,4}^{\underline{1}}}\!\!\!\zeta_i(t)=\!\!\!\!\sum\limits_{i\in\VV_{1,4}^{\underline{1}}}\!\!\!\mathsf{r}^i-\!\!\!\sum\limits_{i\in\VV_{1,4}^{\underline{1}}}\!\!\!x^i(t)&+\!\!\!\!\sum\limits_{i\in\VV_{1,4}^{\underline{1}}}\!\!\!\! \int_{0}^t\!(f^i(\tau)+\dout^i\, g^i(\tau))\textup{d}\tau\\
    &-\!\!\sum\nolimits_{i\in\VV_{1,4}^{\underline{1}}}\!\!\! \beta^i.
\end{align*}
Therefore, we can write
\begin{align*}
   n_{2,3}\,\mu(t)=&\!\!\!\!\!\!\!\sum_{j\in(\mathcal{V}_{1,2}^{\underline{1}}\cup\mathcal{V}_{1,3}^{\underline{1}})}\!\!\!\!\!\! \mathsf{r}^i~-\!\!\!\!\!\sum_{j\in(\mathcal{V}_{1,2}^{\underline{1}}\cup\mathcal{V}_{1,3}^{\underline{1}})}\!\!\!\!\!\! x^i(t)\,-\!\!\!\!\!\sum\limits_{j\in(\mathcal{V}_{1,2}^{\underline{1}}\cup\mathcal{V}_{1,3}^{\underline{1}})} \!\!\!\!\!\beta^i\\
   &+\!\!\!\!\!\!\!\sum\limits_{j\in(\mathcal{V}_{1,2}^{\underline{1}}\cup\mathcal{V}_{1,3}^{\underline{1}})}\!\!\! \int_{0}^t\!(f^j(\tau)+\dout^j\, g^j(\tau))\textup{d}\tau+n_{2,3}\,x^1(t). 
\end{align*}
The proof then follows from the necessary  condition~\eqref{eq::local_beta-i} on the obfuscation signals, and the fact that $\lim_{t\to\infty}n_{2,3}\,x^1(t)-\sum_{j\in(\mathcal{V}_{1,2}^{\underline{1}}\cup\mathcal{V}_{1,3}^{\underline{1}})}\! x^i(t)=0$ (recall that $\lim_{t \to \infty}x^i(t)=\lim_{t \to \infty}x^j(t),~\forall i,j \in \VV $).
\end{proof}

\subsection{Case 2 and Case 3 knowledge sets}
The first result below shows that if $\beta^i$ corresponding to the locally chosen admissible obfuscation signals of an agent $i\in\VV$ is not known to the eavesdropper, the privacy of the agent $i$ is preserved even if the eavesdropper knows all the transmitted input and output signals of agent $i$ and the parameter $\alpha$. The proof of this lemma is given in the appendix. 

\begin{lem}[Privacy preservation for $i\in\VV$ via a concealed $\beta^i$]\label{lem:case2_priv}
Consider the modified static average consensus algorithm~\eqref{eq::consensus-modified} with a set of locally chosen admissible obfuscation signals $\{f^j,g^j\}_{j=1}^N$ over a strongly connected and weight-balanced digraph $\mathcal{G}$.
Let the knowledge set of the eavesdropper $1$ include the form of conditions~\eqref{eq::local_beta-i} and~\eqref{eq::nec-suf-admin-sig}, and also the parameter $\alpha$ that the agents agreed to use. Let agent $1$ be the in-neighbor of agent $i\in\VV$ and all the out-neighbors of agent $i$, i.e., agent $1$ knows $\{y^j(t)\}_{j\in\mathcal{N}_{\textup{out}+i}^i}$, $t\in\real_{\geq0}$. Then, the eavesdropper $1$ can obtain $\mathsf{r}^i$ of agent $i$ if and only if it knows $\beta^i$.
\end{lem}
A similar statement to that of Lemma~\ref{lem:case2_priv} can be made about an external eavesdropper. In the case of the external eavesdropper, it is very likely that the eavesdropper does not know $\alpha$, as well. Building on the result of Lemma~\ref{lem:case2_priv}, we make our final formal privacy preservation statement as follows. 

\begin{thm}[Privacy preservation using the modified average consensus algorithm~\eqref{eq::consensus-modified} when the knowledge set of the eavesdroppers is given by Case 2 in Definition~\ref{def::know}]
\label{thm::case2_priv_main}Consider the modified static average consensus algorithm~\eqref{eq::consensus-modified} with a set of locally chosen admissible obfuscation signals $\{f^j,g^j\}_{j=1}^N$ over a strongly connected and weight-balanced digraph $\mathcal{G}$.
Let the knowledge set of the internal eavesdropper $1$ and the external eavesdropper $\textup{ext}$ be given by Case 2 in Definition~\ref{def::know}. Then, the eavesdropper 1 (resp. agent $\textup{ext}$) cannot reconstruct the reference value $\mathsf{r}^i$ of any agent $i\in\VV\backslash\{1\}$ (resp. $i\in\VV$).
\end{thm}
\begin{proof}
Any agent $i\in\VV\backslash\{1\}$ satisfies either $\mathcal{N}_{\textup{out}+i}^i\subset \mathcal{N}_{\textup{out}+1}^1$ or $\mathcal{N}_{\textup{out}+i}^i\not\subset \mathcal{N}_{\textup{out}+1}^1$. Since the eavesdropper $1$ does not know $\{\beta^i\}_{j=2}^N$, if $\mathcal{N}_{\textup{out}+i}^i\subset \mathcal{N}_{\textup{out}+1}^1$, $i\in\VV\backslash\{1\}$, (agent $1$ has access to all the transmitted input and output signals of agent $i$), it follows from Lemma~\ref{lem:case2_priv} that it cannot reconstruct $\mathsf{r}^i$. Consequently, if $\mathcal{N}_{\textup{out}+i}^i\not\subset \mathcal{N}_{\textup{out}+1}^1$, $i\in\VV\backslash\{1\}$, since the eavesdropper $1$ lacks more information (it does not have access to some or all of the transmitted input and output signals of agent $i$), we conclude that the eavesdropper $1$ cannot reconstruct $\mathsf{r}^i$.
The proof of the statement for the external eavesdropper is similar to that of the internal eavesdropper $1$, and is omitted for brevity (note here that the external eavesdropper $\text{ext}$ lacks the knowledge of $\alpha$, as well).  
\end{proof}
Next we show that in fact, knowing $\beta^i$, $i\in\VV$, e.g., when it is known that agents use $\beta^i=0$, does not result in the breach of privacy against external eavesdroppers that do not know $\alpha$.
\begin{thm}[Privacy preservation using the modified average consensus algorithm~\eqref{eq::consensus-modified} when the knowledge set of the eavesdroppers is given by Case 3 in Definition~\ref{def::know}]
\label{thm::case2_priv_main}Consider the modified static average consensus algorithm~\eqref{eq::consensus-modified} with a set of locally chosen admissible obfuscation signals $\{f^j,g^j\}_{j=1}^N$ over a strongly connected and weight-balanced digraph $\mathcal{G}$.
Let the knowledge set of the external eavesdropper $\textup{ext}$ be given by Case 3 in Definition~\ref{def::know}. Then, the eavesdropper $\textup{ext}$ cannot reconstruct the reference value $\mathsf{r}^i$ of any agent $i\in\VV$.
\end{thm}
\begin{proof}
 The transmitted out signals of the agents implementing~\eqref{eq::consensus-modified} with the locally chosen admissible obfuscation signals $\{f^j,g^j\}_{j=1}^N$ are
$\vect{y}(t)= \text{e}^{-\LL t}\vect{x}(0)+\int_{0}^{t} \textup{e}^{-\lL (t-\tau)}(\vect{f}(\tau)+\vectsf{A}\,\vect{g}(\tau))\,\textup{d}\tau+\vect{g}(t)$. Now consider an alternative implementation of~\eqref{eq::consensus-modified} with initial conditions $\vect{x}'(0)=\vect{x}(0)-a\vect{1}$ and 
 $\vect{g}'(t)=\vect{g}(t)+a\vect{1}$ and $\vect{f}'(t)=\vect{f}(t)-[\dout^1,\cdots,\dout^N]^{\top}a$ for any $a\in\real$. Note that $\{f^{j'},g^{j'}\}_{j=1}^N$ are valid locally chosen admissible obfuscation signals that satisfy \eqref{eq::local_beta-i} and \eqref{eq::nec-suf-admin-sig-a} with the same parameter $\beta^i$, $i\in\VV$ of $\{f^j,g^j\}_{j=1}^N$ and satisfy \eqref{eq::nec-suf-admin-sig-b} with $\alpha+a$ where $\alpha$ is the parameter of~\eqref{eq::nec-suf-admin-sig-b} corresponding to $\{g^j\}_{j=1}^N$. The transmitted out signal of the agents in this implementation are  $
 \vect{y}'(t)= \text{e}^{-\LL t}\vect{x}'(0)+\int_{0}^{t} \textup{e}^{-\lL (t-\tau)}(\vect{f}'(\tau)+\vectsf{A}\,\vect{g}'(\tau))\,\textup{d}\tau+\vect{g}'(t)=\text{e}^{-\LL t}\vect{x}(0)-a\text{e}^{-\LL T}\vect{1}+\int_{0}^{t} \textup{e}^{-\lL (t-\tau)}(\vect{f}(\tau)+\vectsf{A}\,\vect{g}(\tau)-[\dout^1,\cdots,\dout^N]^{\top}a+\vectsf{A}\vect{1}a)\,\textup{d}\tau+\vect{g}'(t)=\vect{x}(t)-a\text{e}^{-\LL t}\vect{1}+\vect{g}(t)+a\vect{1}=\vect{x}(t)+\vect{g}=\vect{y}(t)$, where we used $\text{e}^{-\LL t}\vect{1}=\vect{1}$. Since the eavesdropper $\text{ext}$ does not know the parameter of the condition~\eqref{eq::nec-suf-admin-sig-b} and $\vect{y}'(t)\equiv\vect{y}(t)$ for all $t\in\real_{\geq0}$, it cannot distinguish between the actual and the alternative implementations. Therefore, it cannot uniquely identify the initial condition of the agents. 
\end{proof}

\begin{rem}[Guaranteed privacy preservation when an ultimately secure authority assigns the admissible obfuscation signals]{\rm
 If there exists an ultimately secure and trusted authority that assigns the agents' admissible private obfuscation signals in a way that they collectively satisfy~\eqref{eq::nec-suf-admin-col}, the privacy of the agents is not trivially guaranteed. This is because it is rational to assume that the eavesdroppers know the necessary condition~\eqref{eq::nec-suf-admin-col} and may be able to exploit it to their benefit. However, in light of Theorem~\ref{thm::case2_priv_main}, we are now confident to offer the privacy preservation guarantee for such a case. This is because in this case, the eavesdroppers' knowledge set lacks more information than Case 2 in Definition~\ref{def::know} (note that the locally chosen admissible obfuscation signals are a specially structured subset of all the possible classes of the admissible obfuscation signals).}
\end{rem}

\section{Performance Demonstration}\label{sec::numeric}
\subsection{Stochastic vs. deterministic privacy preservation}
\begin{figure}
     \centering
\includegraphics[height=1in,width=0.37\textwidth]{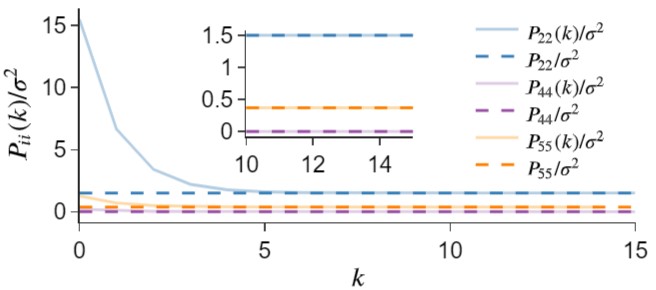}\vspace{-0.1in}
        \caption{{\small Agent 1's (eavesdropper) maximum likelihood estimator's result when method of~\cite{mo2017privacy} is used over graph of Fig.~\ref{fig::k-anom-num}(a).}}
        \label{fig:mo}
\end{figure}

\begin{figure}
    \centering
    \includegraphics[width=0.48\textwidth]{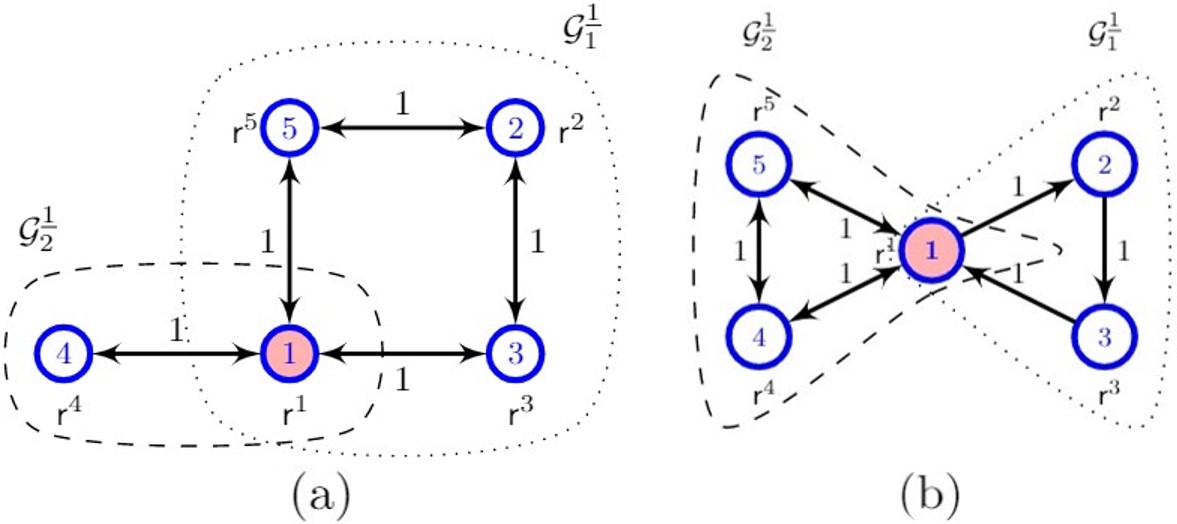}
    \caption{{\small Two strongly connected and weight-balanced graphs $\mathcal{G}$.}
    }
    \label{fig::k-anom-num}
\end{figure}

The deterministic and stochastic approaches to privacy preservation withhold different definitions of a private agent. In our deterministic setup, privacy is preserved when an eavesdropper, despite its knowledge set, ends up in an underdetermined system of equations when it wants to obtain the reference value of an agent. Therefore, the eavesdropper is left with infinite guesses of a private agent's reference value, which it cannot favor any of them more than the other. However, the stochastic privacy of an agent is preserved when the eavesdropper's estimate of the reference value yields a non-zero uncertainty. For example, in~\cite{mo2017privacy} a maximum likelihood estimator is used by the eavesdropper to estimate the reference value of the other agents. It is shown that the variance of $P(k)$ of this estimator converges to a constant matrix $P$. The privacy statement determines that the agents' privacy whose corresponding component in $P$ converges to zero is breached. More specifically, given a vector $\zeta$, a space of the agents' initial condition $\zeta^\top \vect{x}(0)$ is disclosed to the eavesdropper if $\zeta^\top P \zeta=0$ and if $\zeta^\top P \zeta>0$, it is interpreted as conserving the privacy of the subspace. In this setting, for an agent whose corresponding component of $P$ is non-zero, the eavesdropper does not know the agent's exact reference value, but it has an estimate on it. Hence, we tend to favor the deterministic notion of privacy over stochastic as the deterministic approach reveals less information. Figure~\ref{fig:mo} is the replicate of the result of an example study over the graph in Fig.~\ref{fig::k-anom-num}(a) in~\cite{mo2017privacy}, which shows the evolution of the covariance of the maximum-likelihood estimator of the eavesdropper. As expected $P_{44}$  converges to zero but $P_{22}$ and $P_{55}$ not. Even though $P_{22}$ and $P_{55}$ are non-zero,  they are pretty small, indicating that the eavesdropper can have a good estimate of the reference values of these agents. In contrast in our work, our privacy preservation shows that for agents whose privacy is preserved, the eavesdropper not only cannot obtain the reference value but also cannot establish an estimate. 

Consider the network given in Fig~\ref{fig::k-anom-num}(a). To demonstrate over results consider the following  three implementations of the modified continuous-time Laplacian average consensus algorithm~\eqref{eq::consensus-modified} with the reference values and the additive obfuscation signals as follows:
  {\small \begin{align*} 
 &\text{Case (1)}:~~\vect{r}=[-3,5,1,-2,10]^\top,\\
 &\vect{f}(t)=[-3,-2\textup{cos}( \frac{t}{t^2+1}), \frac{t^5}{t^5+1},\textup{tan}(\frac{\pi}{4}\textup{tanh}(t)),-2\textup{tanh}(t)]^\top\!\!\!,
 \\&
 \vect{g}(t)=[1+0.23\textup{e}^{-t},\textup{cos}(10\pi \frac{t^2}{t^5+1}),(1+\textup{e}^{-t}\textup{sin}(10t))\textup{tanh}(t),\\
&~\qquad\quad 1+\textup{e}^{(t-1)^2},\textup{log}(\textup{e}-\textup{e}^{0.1t}(1+\textup{sin}(t)))]^\top\!\!  .
\\~\\
&\text{Case (2)}:~~\vect{r}=[-3,15,-4,-2,5]^\top,\\
&\vect{f}(t)=[-3,-2\textup{cos}( \frac{t}{t^2+1}),-10\textup{e}^{-2t}+\frac{t^5}{t^5+1},\textup{tan}(\frac{\pi}{4}\textup{tanh}(t)),\\ &~\qquad\quad -10\textup{e}^{-2t}-2\textup{tanh}(t)]^\top\!\!\!,
\\&
\vect{g}(t)=[1+0.23\textup{e}^{-t},\textup{cos}(10\pi \frac{t^2}{t^5+1}),\\&~\qquad\quad 
5\textup{e}^{-2t}+(1+\textup{e}^{-t}\textup{sin}(10t))\textup{tanh}(t),1+\textup{e}^{(t-1)^2},\\&~\qquad\quad  5\textup{e}^{-2t}+\textup{log}(\textup{e}-\textup{e}^{0.1t}(1+\textup{sin}(t)))]^\top\!\!  .
\\~\\
&\text{Case (3)}:~~\vect{r}=[-3,25,-9,-2,0]^\top,\\
&\vect{f}(t)=[-3, \allowbreak-2\textup{cos}( \frac{t}{t^2+1}), -20\textup{e}^{-2t}+\frac{t^5}{t^5+1},\textup{tan}(\frac{\pi}{4}\textup{tanh}(t)),\\&~\qquad\quad -20\textup{e}^{-2t}-2\textup{tanh}(t)]^\top\!\!\!,
\\&
\vect{g}(t)=[1+0.23\textup{e}^{-t},\textup{cos}(10\pi \frac{t^2}{t^5+1}),\\&~\qquad\quad 10\textup{e}^{-2t}+(1+\textup{e}^{-t}\textup{sin}(10t))\textup{tanh}(t),1+\textup{e}^{(t-1)^2},\\&~\qquad\quad 10\textup{e}^{-2t}+\textup{log}(\textup{e}-\textup{e}^{0.1t}(1+\textup{sin}(t)))]^\top\!\!  .
\end{align*}}
Let Case (1) correspond to the actual operation case, and the other two cases be admissible alternative ones. Here, all the admissible obfuscation signals are smooth, uniformly continuous and non-vanishing. They satisfy~\eqref{eq::local_beta-i}, ~\eqref{eq::nec-suf-admin-sig-a} and~\eqref{eq::nec-suf-admin-sig-b} with $\alpha=1$ and $\beta^i=0$, $i\in\VV=\{1,2,3,4,5\}$. The plots in the top row of  Fig.~\eqref{fig:sim} confirms convergence of the algorithm  to the exact average, as guaranteed by Theorem~\ref{eq::consensus-modified}. The plots in the second row of Fig.~\eqref{fig:sim} show that the transmitted-out signal $y^i$ of each agent $i\in\VV$ satisfies $\lim_{t\to\infty}{y^i(t)}=\frac{1}{N}\sum_{j=1}^N\mathsf{r}^j+\alpha$. Let $\delta y^i(t)$  be the communication signals difference between Case ($j$), $j\in\{2,3\}$ and Case (1). As seen in the two bottom plots in Fig~\ref{fig:sim}, only $\delta y^2(t)$ is non-zero. 
 This means that agent $1$, in all three cases, receives exactly the same transmission messages from its neighbors, agents $4$, $5$, and $3$. This result, as predicted by Theorem~\ref{eq::consensus-modified}, shows that agent $1$, the eavesdropper, cannot tell whether $\mathsf{r}^2$ is equal to $5$ of Case (1), $15$ of Case (2) or $25$ Case (3). Moreover, agent $1$ is not able to say which one of these cases is more probable. A similar statement can be made about agent $3$ and $5$ whose privacy is guaranteed in our framework. While the privacy of agent $2$, $3$ and $5$ is preserved, according to Lemma~\ref{lem::easy-prey-external-ad}, agent $1$ can employ an observer of form~\eqref{eq::out-observer} to asymptotically estimate the reference value of agent $4$. The response of this estimator is shown in Fig.~\ref{fig:breach4}. Here to make a comparison study with respect to~\cite{mo2017privacy}, we used the undirected graph of Fig~\ref{fig::k-anom-num}(a).

\subsection{Performance over a digraph with external and internal eavesdroppers}
The first demonstration study we conduct is using execution of the modified static average consensus algorithm~\eqref{eq::consensus-modified} over the strongly connected and weight-balanced digraph in Fig.~\ref{fig::k-anom-num}(b). The reference value and the locally chosen obfuscation signals of the agents~are 
\begin{align}\label{eq::main_numeric}
       & \mathsf{r}^1=3,~ \mathsf{r}^2=2,~\mathsf{r}^3=5,
       \mathsf{r}^4=-3,~ \mathsf{r}^5=-1,\nonumber\\
       &f^l(t)=\dout^l(\sin(l\frac{\pi}{12})+\cos(l\frac{\pi}{12}))\frac{\sqrt{(2\,l)}}{4l}\textup{e}^{-t},\\
       &g^l(t)=\sin(l\frac{\pi}{12}+l\pi t^2),\qquad\qquad\qquad\qquad l\in\VV. \nonumber
\end{align}
The locally chosen admissible obfuscation signals here satisfy the conditions in Theorem~\ref{thm::main_alpha_beta_known} with    $\alpha=0$ and $\beta^i=0$, $i\in\VV$.  The interested reader can examine these conditions conveniently using the online integral calculator~\cite{inegralCalcu}. Let the eavesdropper be agent $1$ whose knowledge set is~\eqref{eq::knowledge_set-1} (Case 1 in Definition~\ref{def::know}). With regards to agents $4$ and $5$, despite use of non-vanishing obfuscation signals $g^4$ and $g^5$, as guaranteed in Lemma~\ref{lem::easy-prey-external-ad},  agent $1$ can employ local observers of the form~\eqref{eq::out-observer} to obtain $x^4(0)=\mathsf{r}^4=-3$ and $x^5(0)=\mathsf{r}^5=-1$ (see Fig.~\ref{fig::breach45}).
Agent $1$ however, cannot uniquely identify $\mathsf{r}^2$ and $\mathsf{r}^3$, since $\mathcal{N}^2_{\textup{out}}=\{3\}\not\subseteq\mathcal{N}^1_{\textup{out}+1}=\{1,2,4,5\}$. To show this, consider an \emph{alternative} implementation of algorithm~\eqref{eq::consensus-modified} with initial conditions and admissible obfuscation signals 
    \begin{align}
       & {{x}^1}'\!(0)\!=\!3,\,\, {{x}^2}'\!(0)\!=\!1,\,\,{{x}^3}'\!(0)\!=\!6,\,\,{{x}^4}'\!(0)\!=\!-3,\,\,{{x}^5}'\!(0)\!=\!-1,\nonumber\\
       &{f^i}'(t)\!=\!f^i(t),\quad\quad\quad\quad{g^i}'(t)\!=\!{g}^i(t),~~\quad \quad i\in\{1,3,4,5\},\nonumber\\
       &{f^2}'(t)\!=\!f^2(t)-\textup{e}^{-t},\quad{g^2}'(t)\!=\!g^2(t)+\textup{e}^{-t},\label{eq::alter_numeric}
    \end{align}
where $\frac{1}{5}\sum_{i=1}^5{{x}^i}'(0)\!=\!\frac{1}{5}\sum_{i=1}^5{x}^i(0)\!=\!\frac{1}{5}\sum_{i=1}^8\mathsf{r}^i\!=\!1.2$. As Fig.~\ref{fig::simulation} shows the execution of algorithm~\eqref{eq::consensus-modified} using the initial conditions and obfuscation signals~\eqref{eq::main_numeric} (the actual case) and those in~\eqref{eq::alter_numeric} (an alternative case) converge to the same final value of $1.2$.  
Let $\delta y^i=y^i-{{y}^i}'$, $i\in\until{5}$ be the error between the output of the agents in the actual and the alternative cases. As Fig.~\ref{fig::simulation} shows $\delta y^i\equiv 0$ for all $i\in\Nout^1=\{2,4,5\}$. This means that agent $1$ cannot distinguish between the actual and the alternative cases and therefore, fails to identify uniquely the initial values of agent $2$ and also agent $3$. Figure~\ref{fig::breach2} shows that an external eavesdropper that has access to the output signals of agents $2$ and its  knowledge set is~\eqref{eq::knowledge_set-1} can employ an observer of the form~\eqref{eq::out-observer-external} to identify the initial value of agent $2$, i.e., $\mathsf{r}^2=2$.


\begin{figure}
     \centering
     \includegraphics[width=0.48\textwidth]{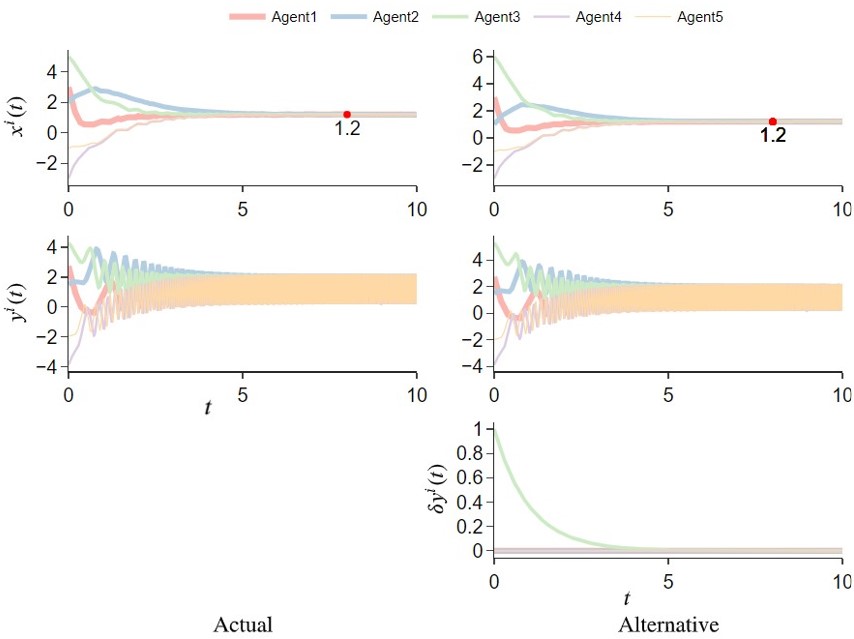}
      \caption{{\small Trajectories of the  state of the agents under the actual initial conditions and the obfuscation signals~\eqref{eq::main_numeric} as well as the alternative ones in~\eqref{eq::alter_numeric} and time history of the difference between the output signal of an agent in actual implementation scenario and its output signal in the alternative implementation described in~\eqref{eq::alter_numeric}.}}
        \label{fig::simulation}
\end{figure}

\begin{figure}
     \centering
     \includegraphics[width=0.40\textwidth]{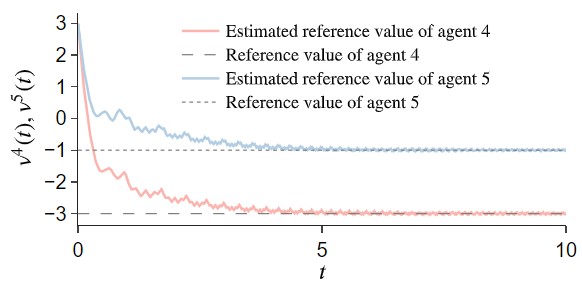}
      \caption{{\small Time history of the observers of the form~\eqref{eq::out-observer} that agent $1$ with knowledge set~\eqref{eq::knowledge_set-1} uses to obtain $\mathsf{r}^4$ and $\mathsf{r}^5$.}}
        \label{fig::breach45}
\end{figure}

\begin{figure}
     \centering
     \includegraphics[width=0.40\textwidth]{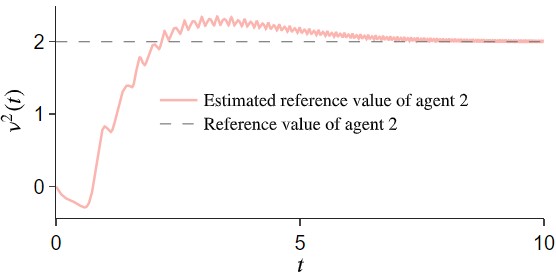}
      \caption{{\small Time history of the observer~\eqref{eq::out-observer-external} of an external eavesdropper with knowledge set~\eqref{eq::knowledge_set-1} that wants to obtain $\mathsf{r}^2$ and has direct access to $y^2$ and $y^3$ for all $t\in\real_{\geq0}$.}}
        \label{fig::breach2}
\end{figure}

\begin{figure}
     \centering
     \includegraphics[width=0.48\textwidth]{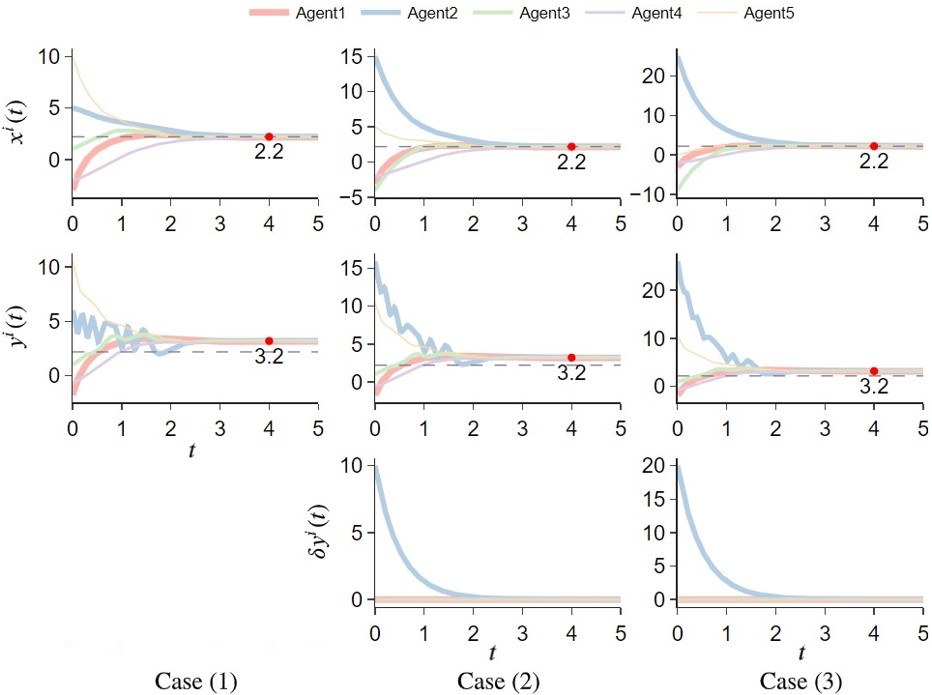}
      \caption{{\small The consensus results for 3 different cases.}}
        \label{fig:sim}
\end{figure}

\begin{figure}
     \centering
     \includegraphics[width=0.4\textwidth]{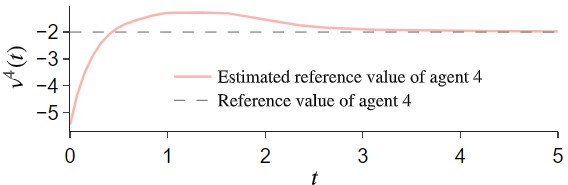}
      \caption{{\small Privacy breach of agent 4 in all 3 cases of~\ref{fig::k-anom-num}(a).}}
        \label{fig:breach4}
\end{figure}

\section{Conclusions}
In this paper, we considered the problem of preserving the privacy of the reference value of the agents in an average consensus algorithm using additive obfuscation signals. We started our study by characterizing the set of the necessary and sufficient conditions on the admissible obfuscation signals, which do not perturb the final convergence point of the algorithm.
We assessed the privacy preservation property of the average consensus algorithm with the additive obfuscation signals against internal and external eavesdroppers, depending on how much knowledge the eavesdroppers have about the necessary conditions that specify the class of signals that the agents choose their local admissible obfuscation signals from. We showed that if the necessary conditions are fully known to the eavesdroppers, then an internal or external eavesdropper that has access to all the transmitted input and out signals of an agent can employ an asymptotic observer to obtain the reference value of that agent. Next, we showed that indeed having access to all the transmitted input and out signals of an agent at all $t\in\real_{\geq0}$ is the necessary and sufficient condition for an eavesdropper to identify the initial value of that particular agent. On the other hand, we showed that if the necessary conditions defining the locally chosen admissible obfuscation signals are not fully known to the eavesdroppers, then the eavesdroppers cannot reconstruct the reference value of any other agent in the network. Our future work includes extending our results to other multi-agent distributed algorithms such as dynamic average consensus and distributed optimization algorithms.


\section{Appendix}\label{sec::appen}
To provide proofs for our lemmas and theorems we rely on a set of  auxiliary results, which we state first.

\begin{lem}[Auxiliary result 1]\label{lem::nec-suf-hurwitz-local}
Let $\lL$ be the Laplacian matrix of a strongly connected and weight-balanced digraph. Recall $\lL^{+}=\rR^\top\lL\rR$ from~\eqref{eq::lL+}. Let  $\vect{g}(t)=[g_1(t),...,g_n(t)]^\top\in\mathcal{L}^{\infty}_{n}$.
Then,  
\begin{align}\label{eq::conv-equ-a}
    & \lim_{t\to\infty}\int_{0}^{t} \textup{e}^{-\lL^{+}(t-\tau)}\rR^\top\lL\,\vect{g}(\tau)\textup{d}\tau=\vect{0},
    \end{align}
   is guaranteed to hold if and only if
    \begin{align}
    &\lim_{t\to\infty}\int_{0}^{t}\textup{e}^{-(t-\tau)} g^i(\tau)\,\textup{d}\tau=\alpha\in\real,\quad i\in\until{N}.\label{eq::conv-equ-b}
\end{align}
\end{lem}
\begin{proof}
Let
\begin{align}
\dvect{\zeta}&=-\lL^{+}\vect{\zeta}+\rR^\top\lL\vect{g}(t),\quad \quad\qquad \vect{\zeta}(0)\in\real^{N-1},\label{eq::lin-eq-a}\\
\dot{\vect{\eta}}&=-\vect{\eta}+\rR^\top\lL\vect{g}(t),~\qquad\qquad\quad \vect{\eta}(0)\in\real^{N-1}\label{eq::lin-eq-b}.
\end{align}
The trajectories $t\mapsto\vect{\zeta}$ and $t\mapsto\vect{\eta}$ of these two dynamics for $t\in\real_{\geq0}$ are given by
\begin{align}
\vect{\zeta}(t)&=\textup{e}^{-\lL^{+}t}\vect{\zeta}(0)+\int_{0}^{t} \textup{e}^{-\lL^{+}(t-\tau)}\rR^\top\lL\vect{g}(\tau)\textup{d}\tau,\label{eq::lin-time-a}\\
\vect{\eta}(t)&=\textup{e}^{-t}\vect{\eta}(0)+\rR^\top\lL\int_{0}^{t} \textup{e}^{-(t-\tau)}\,\vect{g}(\tau)\textup{d}\tau\label{eq::lin-time-b}.
\end{align}
Let $\vect{e}=\vect{\zeta}-\vect{\eta}$. Then, the error dynamics between~\eqref{eq::lin-eq-a} and~\eqref{eq::lin-eq-b} is given by
\begin{align}\label{eq::err-lin1}
    \dot{\vect{e}}=-\vect{e}+(\vect{I}-\lL^{+})\vect{\zeta}.
\end{align}
or equivalently
\begin{align}\label{eq::err-lin2}
    \dot{\vect{e}}=-\lL^{+}\vect{e}+(\lL^{+}+\vect{I})\vect{\eta}.
\end{align}
Let~\eqref{eq::conv-equ-a} hold. Since $-\lL^{+}$ is a Hurwitz matrix, we have $\lim_{t\to\infty}\vect{\zeta}(t)=0$. Moreover, since $\vect{g}$ is essentially bounded, the trajectories of $\vect{\zeta}$ are guaranteed to be bounded. Therefore, considering error dynamics~\eqref{eq::err-lin1}, by invoking the ISS stability results~\cite{SND-DVE-EDS:11}, we have the guarantees that $\lim_{t\to\infty}\vect{e}(t)=\vect{0}$, and consequently $\lim_{t\to\infty}\vect{\eta}(t)=\vect{0}$. As such, from~\eqref{eq::lin-time-b} we~obtain
\begin{align}\label{eq::B-lim}
\rR^\top\lL \lim_{t\to\infty}\int_{0}^{t} \textup{e}^{-(t-\tau)}\vect{g}(\tau)\textup{d}\tau=\vect{0}.
\end{align}
The nullspace of $\rR^\top\lL\!\in\!\real^{(N-1)\times N}$ is spanned by~$\vect{1}_N$,~thus,
$$\lim_{t\to\infty}\int_{0}^{t} \textup{e}^{-(t-\tau)}\vect{g}(\tau)\textup{d}\tau=\alpha \vect{1}_N,\quad\alpha\in\real,$$
which validates~\eqref{eq::conv-equ-b}. Now let~\eqref{eq::conv-equ-b} hold. Then, using~\eqref{eq::lin-time-b}, we obtain $\lim_{t\to\infty}\vect{\eta}(t)= \vect{0}$. Since $\vect{g}$ is essentially bounded, the trajectories of $\vect{\zeta}$ are guaranteed to be bounded. Thereby, considering error dynamics~\eqref{eq::err-lin2}, by invoking the ISS stability results~\cite{SND-DVE-EDS:11}, we have the guarantees that $\lim_{t\to\infty}\vect{e}(t)=\vect{0}$, and consequently $\lim_{t\to\infty}\vect{\eta}(t)=\vect{0}$. Since $-\lL^{+}$ is a Hurwitz matrix, we obtain~\eqref{eq::conv-equ-a} from~\eqref{eq::lin-time-a}.
\end{proof}

\begin{lem}[Auxiliary result 2]\label{lem::conv_vanishing}
Let $\vect{u}:\real_{\geq 0}\to\real^n$ be~an~essentially bounded signal and $\vect{E}\in\real^{n\times n}$ be a Hurwitz~matrix.
\begin{itemize}
    \item[(a)] If $\lim_{t\to\infty}\vect{u}(t)=\bar{\vect{u}}\in\real^{n}$, and $\vect{E}\in\real^{n\times n}$, then 
\begin{align}\label{eq::conv_vanishing}
\lim_{t\to\infty}\int_{0}^t\textup{e}^{\vect{E}\,(t-\tau)}\vect{u}(\tau)\textup{d}\tau=-\vect{E}^{-1}\,\bar{\vect{u}}.
\end{align}

\item[(b)] 
If $\lim_{t\to\infty}\int_{0}^t \vect{u}(\tau)\textup{d}\tau=\bar{\vect{u}}\in\real^{n}$,  then 
\begin{align}\label{eq::conv_integral_bounded}
\lim_{t\to\infty}\int_{0}^t\textup{e}^{\vect{E}\,(t-\tau)}\vect{u}(\tau)\textup{d}\tau=\vect{0}.
\end{align}
\end{itemize}
\end{lem}
\begin{proof}
To prove statement (a) we proceed as follows. Let $\vect{\mu}(t)=\vect{u}(t)-\bar{\vect{u}}$. Next, consider $\dvect{\zeta}=\vect{E}\,\vect{\zeta}+\vect{\mu}$, $\vect{\zeta}(0)\in\real^n$, which gives  $\vect{\zeta}(t)=\textup{e}^{\vect{E}\, t}\vect{\zeta}(0)+\int_{0}^t\textup{e}^{\vect{E}(t-\tau)}\vect{\mu}(\tau)\textup{d}\tau$, $t\geq0$. Since $\vect{E}$ is Hurwitz and $\vect{\mu}$ is an essentially bounded and vanishing signal, by virtue of the ISS results for linear systems~\cite{SND-DVE-EDS:11} we have $\lim_{t\to\infty}\vect{\zeta}(t)=0$. Consequently,  $\lim_{t\to\infty}\int_{0}^t\textup{e}^{\vect{E}\,(t-\tau)}\vect{\mu}(\tau)\textup{d}\tau=\vect{0}$, which guarantees  ~\eqref{eq::conv_vanishing}.

To prove statement (b) we proceed as follows. Consider
\begin{align*}
    \dvect{\zeta}=\vect{u}, ~\dvect{\eta}=\vect{E}\vect{\eta}+\vect{u}, \quad \vect{\zeta}(0)=\vect{0},~\vect{\eta}(0)\in\real^n,
\end{align*}
which result in $\vect{\zeta}(t)=\int_{0}^t \vect{u}(\tau)\textup{d}\tau$ and 
\begin{align}\label{eq::aux_lem2_(b)_eta}
    \vect{\eta}(t)=\textup{e}^{\vect{E}\,t}\vect{\eta}(0)+\!\!\int_{0}^{t}\!\!\textup{e}^{\vect{E}\,(t-\tau)}\vect{u}(\tau)\textup{d}\tau.
\end{align}
 Given the conditions on $\vect{u}$ both $\vect{\zeta}$ and $\vect{\eta}$ are essentially bounded signals (recall that $\vect{E}$ is Hurwitz). Let $\vect{e}=\vect{\eta}-\vect{\zeta}$. Therefore, we can write
\begin{align*}
    \dvect{e}=\vect{E}\,\vect{e}+\vect{E}\,\vect{\zeta},\quad \vect{e}(0)=\vect{\eta}(0)\in\real^n.
\end{align*}
Since $\vect{\zeta}$ is essentially bounded and satisfies $\lim_{t\to\infty}\vect{E}\vect{\zeta}(t)=\vect{E}\bar{\vect{u}}$, with an argument similar to that of the proof of statement (a), we can conclude that $\lim_{t\to\infty}\vect{e}(t)=-\bar{\vect{u}}$. As a result $\lim_{t\to\infty}\vect{\eta}(t)=\vect{0}$. Consequently, from~\eqref{eq::aux_lem2_(b)_eta}, we obtain~\eqref{eq::conv_integral_bounded}.
\end{proof}


\begin{lem}[Auxiliary result 3]\label{lem::subgraph-balanced}
Let $\GG$ be a strongly connected and weight-balanced digraph. Then, every island of any agent $i$, is strongly connected and weight-balanced. 
\end{lem}
\begin{proof}
Without loss of generality, we prove our argument by showing that the island $\GG_{1}^{\underline{1}}$ of agent $1$ is strongly connected and weight-balanced. By construction, we know that there is a directed path from every agent to every other agent in $\GG_{1}^{\underline{1}}$, therefore, $\GG_{1}^{\underline{1}}$ is strongly connected. Next we show that $\GG_{1}^{\underline{1}}$ is weight-balanced. Let $\VV_2=\VV_{1}^{\underline{1}}\backslash\{1\}$ and $\VV_3=\VV\backslash\VV_2$.
Let the nodes of $\GG$ be labeled in accordance to $(1,\VV_2,\VV_3)$, respectively, and partition the graph Laplacian $\lL$ accordingly as
\begin{align*}
    \lL=\begin{bmatrix}\dout^1&-\vectsf{A}_{12}&-\vectsf{A}_{13}\\
    -\vectsf{A}_{21}&\vectsf{L}_{22}&\vect{0}\\
    -\vectsf{A}_{31}&\vect{0}&\vectsf{L}_{33}
    \end{bmatrix}.
\end{align*}
Since $\GG$ is strongly connected and weight-balanced, we have $\lL\vect{1}_N=\vect{0}$ and $\vect{1}_N^\top\lL=\vect{0}$, which guarantee that
\begin{align}\label{eq::island_SCWB_1}
&\vect{1}_{|\VV_{1}^{\underline{1}}|}^\top \begin{bmatrix}-\vectsf{A}_{12}\\ \vectsf{L}_{22}
    \end{bmatrix}=\vect{0},\qquad  \begin{bmatrix}-\vectsf{A}_{21} & \vectsf{L}_{22}
    \end{bmatrix}\vect{1}_{|\VV_{1}^{\underline{1}}|}=\vect{0}.
\end{align}
Therefore, 
\begin{align*}
&\vect{1}_{|\VV_{1}^{\underline{1}}|}^\top \begin{bmatrix}-\vectsf{A}_{12}\\ \vectsf{L}_{22}
    \end{bmatrix}\vect{1}_{|\VV_{1}^{\underline{1}}|}=0,\quad 
&\vect{1}_{|\VV_{1}^{\underline{1}}|}^\top \begin{bmatrix}-\vectsf{A}_{21} & \vectsf{L}_{22}
    \end{bmatrix}\vect{1}_{|\VV_{1}^{\underline{1}}|}=0,
\end{align*}
which we can use to conclude that $\texttt{sum}(\vectsf{A}_{12}^\top)=\texttt{sum}(\vectsf{A}_{21})$. Let the Laplacian matrix of $\GG_{1}^{\underline{1}}$ be $\lL_1^{\underline{1}}$.
Partitioning this matrix according to order node set $(1,\VV_2)$, we obtain
\begin{align*}
    \lL_1^{\underline{1}}=\begin{bmatrix}{\mathsf{d}}_{\textup{out}}^{1,1}&-\vectsf{A}_{12}\\
    -\vectsf{A}_{21}&\vectsf{L}_{22}
    \end{bmatrix},
\end{align*}
where $\mathsf{d}_{\textup{out}}^{1,1}=\sum_{j\in\VV_2}\mathsf{a}_{1j}=\texttt{sum}(\vectsf{A}_{12}^\top)$.
To establish $\GG_{1}^{\underline{1}}$ is weight-balanced digraph, we show next that $\vect{1}_{|\VV_{1}^{\underline{1}}|}^\top\lL_1^{\underline{1}}=\vect{0}$. From $\vect{1}_N^\top\lL=\vect{0}$, it follows that $\vect{1}_{|\VV_{1}^{\underline{1}}|}^\top\begin{bmatrix}-\vectsf{A}_{12}\\ \vectsf{L}_{22}
    \end{bmatrix}=\vect{0}$.  Therefore, to prove $\GG_{1}^{\underline{1}}$ is weight-balanced, we need to show that $\mathsf{d}_{\textup{out}}^{1,1}+\texttt{sum}(-\vect{A}_{21})=0$, which follows  immediately from $\mathsf{d}_{\textup{out}}^{1,1}=\texttt{sum}(\vectsf{A}_{12}^\top)$ and  $\texttt{sum}(\vectsf{A}_{12}^\top)=\texttt{sum}(\vectsf{A}_{21})$.
\end{proof}
\bigskip

Next we present the proof of our main results.

\medskip

\begin{proof}[Proof of Theorem~\ref{thm::main-col}]
To prove necessity, we proceed as follows. We write the algorithm~\eqref{eq::consensus-modified} in compact form 
\begin{align}\label{eq::compact_modified_Con}
\dvect{x}=-{\lL}\,\vect{x}-{\lL}\,\vect{g}+\vect{f}+\Dout\,\vect{g}=-{\lL}\,\vect{x}+\vect{f}+\vect{A}\,\vect{g}.
\end{align}

Left multiplying both sides of~\eqref{eq::compact_modified_Con} by $\vect{1}_N^{\top}$ gives
\begin{align*}
    \sum\nolimits_{j=1}^N\dot{x}^j(t)=\sum\nolimits_{j=1}^N(f^i(t)+\dout^i\, g^i(t)),
\end{align*}
which results in
\begin{align*}
    \sum\nolimits_{j=1}^N\!\!\!{x}^j(t)\!=\!\sum\nolimits_{j=1}^N\!\!{x}^j(0)\!+\!\!\!\int_{0}^t\!\!\sum\nolimits_{j=1}^N(f^i(\tau)+\dout^i\, g^i(\tau))\,\textup{d}\tau.
\end{align*}
Because $x^i(0)=\mathsf{r}^i$, to ensure $\lim_{t\to\infty} x^i(t)=\avrg{\mathsf{r}}$, $i\in\VV$, we necessarily need~\eqref{eq::nec-suf-admin-col-b}.

Next, we apply the change of variable 
 \begin{align}\label{eq::var_ch}
\vect{p}=\begin{bmatrix}p_1\\ \vect{p}_{2:N}\end{bmatrix}=\,\vect{T}\,\vect{x},
\end{align}
where $\vect{T}$ is defined in~\eqref{eq::lL+}, to write~\eqref{eq::compact_modified_Con} in the equivalent form
\begin{subequations}\label{eq::p1p2}
\begin{align}
\dot{p}_1&\!=\!\frac{1}{\sqrt{N}}\sum\nolimits_{i=1}^N (f^i+\dout^i\, g^i),\label{eq::p1p2-p1}\\
\dvect{p}_{2:N}&\!=\!-\lL^{+}  \,\vect{p}_{2:N}\!+\rR^\top(\vect{f}+\vect{A}\,\vect{g}).\label{eq::p1p2-p2}
\end{align}
\end{subequations}
 The solution of~\eqref{eq::p1p2} is 
\begin{subequations}\label{eq::equiv_traject}
\begin{align}
{p}_1(t)=\,&\frac{1}{\sqrt{N}}\sum\nolimits_{i=1}^N\!\!x^i(0)+ \label{eq::p1-time}\\
&\quad\frac{1}{\sqrt{N}}\int_{0}^{t}\sum\nolimits_{i=1}^N (f^i(\tau)+\dout^i\, g^i(\tau))\textup{d}\tau,  \nonumber\\
\vect{p}_{2:N}(t)=\,&\textup{e}^{-\lL^{+}\,t}\,\vect{p}_{2:N}(0)\,+
\nonumber\\
&\qquad\int_{0}^t\!\!\textup{e}^{-\lL^{+} (t-\tau)}\rR^{\top}(\vect{f}(\tau)+\vect{A}\,\vect{g}(\tau))\,\textup{d}\tau. \label{eq::p2-time}
\end{align}
\end{subequations}
Given~\eqref{eq::nec-suf-admin-col-a},~\eqref{eq::p1-time} results in
$\lim_{t\to\infty} p_1(t)=\frac{1}{\sqrt{N}}\sum\nolimits_{i=1}^N\!\!x^i(0)=\avrg{\mathsf{r}}$.
Consequently, given~\eqref{eq::var_ch}, to ensure $\lim_{t\to\infty} x^i(t)=\avrg{\mathsf{r}}$, $i\in\VV$, we need
\begin{align}\label{eq::p2_zero}
    \lim_{t\to\infty}\vect{p}_{2:N}(t)=\vect{0}.
\end{align}
Because for a strongly connected and weight-balanced digraph, $-\lL^{+}$ is a Hurwitz matrix, $\lim_{t\to\infty} \textup{e}^{-\lL^{+}\,t}\vect{p}_{2:N}(0)=\vect{0}$. Then, the necessary condition for~\eqref{eq::p2_zero} is~\eqref{eq::nec-suf-admin-col-b}.

The sufficiency proof follows from noting that under~\eqref{eq::nec-suf-admin-col}, the trajectories of~\eqref{eq::equiv_traject} satisfy $\lim_{t\to\infty}p_1(t)=\frac{1}{\sqrt{N}}\sum\nolimits_{i=1}^N\!x^i(0)$ and $\lim_{t\to\infty}\vect{p}_{2:N}(t)=\vect{0}$. Then, given~\eqref{eq::var_ch} and $x^i(0)=\mathsf{r}^i$ we obtain $\lim_{t\to\infty}x^i(t)=\frac{1}{N}\sum\nolimits_{j=1}^N\!\mathsf{r}^j$, $i\in\VV$.
\end{proof}

\smallskip
\begin{proof}[Proof of Theorem~\ref{thm::main}]
Given~\eqref{eq::local_beta-i}, it is straightforward to see that~\eqref{eq::nec-suf-admin-sig-a} is necessary and sufficient for~\eqref{eq::nec-suf-admin-col-a}. Next,
 we observe that 
using~\eqref{eq::local_beta-i}, we can write
$\lim_{t\to\infty}\int_{0}^t\!\rR^\top\,(\vect{f}(\tau)+\Dout\,\vect{g}(\tau))\textup{d}\tau=\rR^\top\begin{bmatrix}\beta^1&\cdots&\beta^N\end{bmatrix}^\top$. Then, it follows from the statement (b) of Lemma~\ref{lem::conv_vanishing} that 
$\lim_{t\to\infty}\int_{0}^t\!\!\textup{e}^{-\lL^{+} (t-\tau)} \rR^\top\,(\vect{f}(\tau)+\Dout\,\vect{g}(\tau))\textup{d}\tau=\vect{0}$. As a result, given $\vect{f}+\vect{A}\,\vect{g}=\vect{f}+\Dout\,\vect{g}-{\lL}\,\vect{g}$, we obtain
\begin{align}\label{eq::nes-suf-local3-tem}
&\lim_{t\to\infty}\int_{0}^{t}\!\!\textup{e}^{-\lL^{+} (t-\tau)}\rR^{\top}(\vect{f}(\tau)+\vectsf{A}\,\vect{g}(\tau))\,\textup{d}\tau=\nonumber\\
&\quad\quad\qquad\qquad-\lim_{t\to\infty}\int_{0}^{t}\!\!\textup{e}^{-\lL^{+} (t-\tau)} \rR^\top{\lL}\,\vect{g}(\tau)\textup{d}\tau.
\end{align}
Given~\eqref{eq::nes-suf-local3-tem}, by virtue of Lemma~\ref{lem::nec-suf-hurwitz-local},~\eqref{eq::nec-suf-admin-col-b} holds if and only if~\eqref{eq::nec-suf-admin-sig-b} holds.
\end{proof}

\smallskip
\begin{proof}[Proof of Lemma~\ref{lem::admissible-signal}]
When condition (a) holds, the proof of the statement follows from  statement (a) of Lemma~\ref{lem::conv_vanishing}. When condition (b) is satisfied, the proof follows from the statements (a) and (b) of  Lemma~\ref{lem::conv_vanishing} which, respectively, give $\lim_{t\to\infty}\int_{0}^t\! \textup{e}^{-(t-\tau)}g_1(\tau)\textup{d}\tau=\alpha$ and $\lim_{t\to\infty}\int_{0}^t\! \textup{e}^{-(t-\tau)}g_2(\tau)\textup{d}\tau=0$.  
When condition (c) is satisfied, the proof follows from the statement (a) of Lemma~\ref{lem::conv_vanishing} which gives 
$\lim_{t\to\infty}\int_{0}^t\! \textup{e}^{-(t-\tau)}g_1(\tau)\textup{d}\tau=\alpha$
and noting that~$\int_{0}^t \textup{e}^{-(t-\tau)}g_2(\tau)\textup{d}\tau$ is the zero state response of system $\dot{\zeta}=-\zeta+g_2$. Since $g_2(t)$ is essentially bounded, this system is ISS, and as a result it is also integral ISS~\cite{SND-DVE-EDS:11}. Then, $\int_{0}^t \textup{e}^{-(t-\tau)}g_2(\tau)\textup{d}\tau=0$, follows from~\cite[Lemma 3.1]{SND-DVE-EDS:11}.
\end{proof}

\smallskip
\begin{proof}[Proof of Lemma~\ref{lem::flat_output}]
Let the error variables of the two execution of~\eqref{eq::consensus-modified} described in the statement be $\delta x^i(t)={x^i}'(t)-x^i(t)$, $\delta y^i(t)={y^i}'(t)-y^i(t)$, $\delta g^i(t)={g^i}'(t)-g^i(t)$, and $\delta f^i(t)={f^i}'(t)-f^i(t)$, $i\in\VV$. Consequently, 
\begin{subequations}\label{eq::delta_p_q}
\begin{align}
 \!&\delta{x^1}(0)=0,\quad \delta \vect{x}_4=\vect{0},\quad \delta \vect{x}_5(0)=\vect{0},\label{eq::ini_deltax145} \\
  \!&\delta{x^i}(0)\in\real,~~\quad\quad\quad\quad\quad\quad\quad \label{eq::int_delta_1_4_5} i\in({\mathcal{V}}^{\underline{1}}_{1,2}\cup  {\mathcal{V}}^{\underline{1}}_{1,3}),\\
    \! &\delta\vect{x}_2(0)= -\vectsf{A}_{23}\lL_{33}^{-1}\delta\vect{x}_3(0),
    \end{align}
\end{subequations}
and
\begin{subequations}\label{eq::delta_f_g}
\begin{align}
\!&\delta g^1(t)\equiv 0,~~\delta f^1(t)\equiv 0,\label{eq::delta_g1_f1} \\
\!&\delta\vect{g}_l(t)\equiv\vect{0},~~\delta\vect{f}_l(t)\equiv\vect{0},\quad\quad l\in\{3,4,5\},\label{eq::delta_g345_f345}\\
\!\!\!\!\!\!\!\!\!\!&\delta\vect{g}_2(t)\!=\!-\textup{e}^{-\Dout_{22}t}\delta\vect{x}_2(0),~\delta\vect{f}_2(t)\!=\! -\vectsf{A}_{23}\textup{e}^{-\lL_{33}t}\delta\vect{x}_3(0)\label{eq::delta_g2_f2}.
\end{align}
\end{subequations}
Given the inter-agent interactions across the network based on agent grouping in accordance to 
the definition of the island $\mathcal{G}^{\underline{1}}_1$ (see Fig.~\ref{fig::network_island}), the error dynamics pertained to the modified static average consensus algorithm~\eqref{eq::consensus-modified} reads as
\begin{align}\label{eq::island_dynamics}
 &\!\!\! \begin{bmatrix}
  \delta\dot{x}^1\\
  \delta\dot{\vect{x}}_{2}\\
  \delta\dot{\vect{x}}_{3}\\
  \delta\dot{\vect{x}}_{4}\\
  \delta\dot{\vect{x}}_{5}
  \end{bmatrix}\!\!=\!-\!\underbrace{\begin{bmatrix}
  \dout^1 &-\vectsf{A}_{12} &\vect{0} &-\vectsf{A}_{14} &-\vectsf{A}_{15}\\
  -\vectsf{A}_{21}  &\lL_{22} &-\vectsf{A}_{23} &-\vectsf{A}_{24} &\vect{0}\\
  -\vectsf{A}_{31}  &-\vectsf{A}_{32} &\lL_{33} &-\vectsf{A}_{34} &\vect{0}\\
  -\vectsf{A}_{41}  &-\vectsf{A}_{42} &\vect{0} &\lL_{44} &\vect{0}\\
  -\vectsf{A}_{51}  &\vect{0} &\vect{0} &\vect{0} &\lL_{55}
  \end{bmatrix}}_{\lL}\!
  \begin{bmatrix}
  \delta x^1\\\delta\vect{x}_{2} \\ \delta\vect{x}_{3} \\ \delta\vect{x}_{4}\\\delta\vect{x}_{5} 
  \end{bmatrix}\nonumber\\&
  ~+\underbrace{\begin{bmatrix}
  0&\vectsf{A}_{12}&\vect{0}     &\vectsf{A}_{14}&\vectsf{A}_{15}\\
  \vectsf{A}_{21} &\vectsf{A}_{22}&\vectsf{A}_{23}&\vectsf{A}_{24}&\vect{0}\\
  \vectsf{A}_{31} &\vectsf{A}_{32}&\vectsf{A}_{33}&\vectsf{A}_{34}&\vect{0}\\
  \vectsf{A}_{41} &\vectsf{A}_{42}&\vect{0}     &\vectsf{A}_{44}&\vect{0}\\
  \vectsf{A}_{51} &\vect{0}     &\vect{0}     &\vect{0}     &\vectsf{A}_{55}
  \end{bmatrix}}_{\vect{A}}
  \begin{bmatrix}
  \delta{g}^{1}\\\delta\vect{g}_{2} \\ \delta\vect{g}_{3}\\ \delta\vect{g}_{4}\\\delta\vect{g}_{5}
  \end{bmatrix}\!+\!
  \begin{bmatrix}
  \delta{f}^{1}\\\delta\vect{f}_{2} \\ \delta\vect{f}_{3}\\ \delta\vect{f}_{4}\\\delta\vect{f}_{5}
  \end{bmatrix}\!\!.
 \end{align}
Since for a strongly connected and weight-balanced digraph we have $\rank(\lL)=N-1$ and $-(\lL+\lL^\top)\leq 0$, the sub-block matrices $-\lL_{33}$ and $-\lL_{44}$ and $-\lL_{55}$ satisfy $-(\lL_{ii}+\lL_{ii}^\top)<0$, $i\in\until{5}$. Thereby, they are invertible and Hurwitz matrices.

To establish~\eqref{eq::init-alt-sum}, we show $\vect{1}_N^\top\delta \vect{x}(0)=\vect{0}_N$. For this, note that taking into account~\eqref{eq::delta_p_q}, we can write
\begin{align}\label{eq::deltax}
   \delta\vect{x}(0)=\underbrace{\begin{bmatrix}
  0 &\vect{0} &\vect{0} &\vect{0} &\vect{0}\\
  \vect{0}  &\vect{0} &-\vectsf{A}_{23} &\vect{0} &\vect{0}\\
  \vect{0}  &\vect{0} &\lL_{33} &\vect{0} &\vect{0}\\
  \vect{0}  &\vect{0} &\vect{0} &\vect{0} &\vect{0}\\
  \vect{0}  &\vect{0} &\vect{0} &\vect{0} &\vect{0}
  \end{bmatrix}}_{\vect{B}}\begin{bmatrix}0\\\lL_{33}^{-1}\delta\vect{x}_3(0)\\\lL_{33}^{-1}\delta\vect{x}_3(0)\\\vect{0}\\\vect{0}\end{bmatrix}
\end{align}
Comparing $\vect{B}$ with the block partitioned $\lL$ in~\eqref{eq::island_dynamics}, it is evident  that $\vect{1}^\top\vect{B}=\vect{0}$ follows from $\vect{1}^\top\lL=\vect{0}$. Consequently, we can deduce from~\eqref{eq::deltax} that $\vect{1}^\top\delta \vect{x}(0)=0$. Next, given~\eqref{eq::init-alt-sum}, we validate~\eqref{eq::init-alt-converge} by invoking Theorem~\ref{thm::main} and showing that the obfuscation signals $({f^i}',{g^i}')$, $i\in\VV$, satisfy the sufficient conditions in~\eqref{eq::nec-suf-admin-sig}. For $i\in\VV\backslash\VV_{1,2}^{\underline{1}}$, the sufficient conditions in~\eqref{eq::nec-suf-admin-sig} are trivially satisfied. To show~\eqref{eq::nec-suf-admin-sig-a} for  $i\in\VV_{1,2}^{\underline{1}}$, we proceed as follows. First note that since $(f^i,g^i)$, $i\in\VV_{1,2}^{\underline{1}}$, are admissible signals, they necessarily satisfy~\eqref{eq::nec-suf-admin-sig-a}. Next, note that using~\eqref{eq::initial_alter} we can write
\begin{align*}
  \int_{0}^t\big(-\vectsf{A}_{23}\textup{e}^{-\lL_{33}\tau}&\delta\vect{x}_3(0)+\Dout_{22}\textup{e}^{-\Dout_{22}\tau}\delta\vect{x}_2(0)\big)\textup{d}\tau=\\
  &~~ \vectsf{A}_{23}\lL_{33}^{-1}\textup{e}^{-\lL_{33}t}\delta\vect{x}_3(0)-\textup{e}^{-\Dout_{22}\tau}\delta\vect{x}_2(0).
\end{align*}
 Let $\mathfrak{B}_2=[\{\beta^i\}_{i\in\VV_{1,2}^{\underline{1}}}]$. Then, in light of the aforementioned observations and the fact that $-\lL_{33}$ and $-\Dout_{22}$ are Hurwitz matrices we can write
\begin{align*}
 & \lim_{t\to\infty}\int_{0}^{t} \big(\vect{f}_2'(\tau)+\Dout_{22}\vect{g}_2'(\tau)\big)\textup{d}\tau=\\
  &
  \mathfrak{B}_2+\lim_{t\to\infty}\big( \vectsf{A}_{23}\lL_{33}^{-1}\textup{e}^{-\lL_{33}t}\delta\vect{x}_3(0)-\textup{e}^{-\Dout_{22}\tau}\delta\vect{x}_2(0)\big)=\mathfrak{B}_2,
\end{align*}

which shows $({f^i}',{g^i}')$, $i\in\VV_{1,2}^{\underline{1}}$ also satisfy the sufficient condition~\eqref{eq::nec-suf-admin-sig-a}.
Establishing that ${g^i}'$, $i\in\VV_{1,2}^{\underline{1}}$, satisfies the sufficient condition~\eqref{eq::nec-suf-admin-sig-b} follows from admissibility of $g^i$, $i\in\VV_{1,2}^{\underline{1}}$, which ensures it satisfies~\eqref{eq::nec-suf-admin-sig-b}, and direct calculations as show below,
\begin{align*}
\lim_{t\to\infty}\int_{0}^{t}\!\!\!\textup{e}^{-(t-\tau)} {g^i}'(\tau)\,\textup{d}\tau\!&=\\
\!\alpha + \lim_{t\to\infty}&\int_{0}^{t}\!\!\!\textup{e}^{-(t-\tau)} \textup{e}^{-\dout^i \tau}\delta x^i(0)\,\textup{d}\tau\!=\!\alpha.
\end{align*}
Here we used the fact that for a strongly connected digraph we have $\dout^i\geq1$.

To establish~\eqref{eq::unobserv-con} we proceed as follows. We assume that~\eqref{eq::unobserv-con} or equivalently 
\begin{subequations}\label{eq::initial_error_Nout1-alt}
\begin{align}
&\delta {y}^1(t)=\delta{x}^1(t)+\delta {g}^1(t)\equiv \vect{0},\,\qquad t\in\real_{\geq0},\label{eq::delta_y2-alt}\\
&\delta \vect{y}_2(t)=\delta\vect{x}_2(t)+\delta \vect{g}_2(t)\equiv \vect{0},\qquad t\in\real_{\geq0},\label{eq::delta_y2-alt}\\
&\delta \vect{y}_4(t)=\delta\vect{x}_4(t)+\delta \vect{g}_4(t)\equiv \vect{0},\qquad t\in\real_{\geq0},\label{eq::delta_y4-alt}\\
&\delta \vect{y}_5(t)=\delta \vect{x}_5(t)+\delta \vect{g}_5(t)\equiv \vect{0},\qquad t\in\real_{\geq0}.\label{eq::delta_y_5-alt}
\end{align}
\end{subequations}
hold. Then, for the given initial conditions~\eqref{eq::delta_p_q}, we identify the obfuscation signals that make the error dynamics~\eqref{eq::island_dynamics} render such an output. As we show below, these obfuscation signals are exactly the same as~\eqref{eq::delta_f_g}. Then, the proof is established by the fact that given a set of initial conditions and integrable external signals, the solution of any linear ordinary differential equation is unique. That is, if we implement the identified  inputs, the error dynamics is guaranteed to  satisfy~\eqref{eq::initial_error_Nout1-alt}. If~\eqref{eq::initial_error_Nout1-alt} holds, then the error dynamics~\eqref{eq::island_dynamics} reads~as
\begin{subequations}\label{eq::obfuscation_dyno_proof}
\begin{align}
\delta\dot{x}^1&=-\dout^1\delta x^1+\delta f^1,\\
\delta\dvect{x}_2&=-\Dout_{22}\delta\vect{x}_2+\vectsf{A}_{23}\delta\vect{x}_3+\vectsf{A}_{23}\delta\vect{g}_3+\delta\vect{f}_2,\label{eq::obfuscation_dyno_proof-x2}\\
\delta\dvect{x}_3&=-\lL_{33}\delta\vect{x}_3+\vectsf{A}_{33}\delta\vect{g}_3+\delta\vect{f}_3,\\
\delta\dvect{x}_4&=-\Dout_{44}\delta\vect{x}_4+\delta\vect{f}_4,\\
\delta\dvect{x}_5&=-\Dout_{55}\delta\vect{x}_5+\delta\vect{f}_5,
\end{align}
\end{subequations}
Here, we used $\lL_{ii}=\Dout_{ii}-\vectsf{A}_{ii}$, $i\in\{1,2,4,5\}$. Next, we choose the  obfuscation signals according to~\eqref{eq::delta_f_g}. Then,  for the given initial conditions~\eqref{eq::delta_p_q}, we obtain from~\eqref{eq::obfuscation_dyno_proof}, 
\begin{subequations}\label{eq::solution_confirm}
\begin{alignat}{3}
\delta\dot{x}^1&=-\dout^1\delta x^1,&\Rightarrow~& \delta x^1(t)=0~\Rightarrow ~\delta y^1(t)\equiv 0,\\
\delta\dvect{x}_3&=-\lL_{33}\,\delta\vect{x}_3,&\Rightarrow~&\delta\vect{x}_3(t)=\textup{e}^{-\lL_{33}t}\delta\vect{x}_3(0),\\
\delta\dvect{x}_4&=-\Dout_{44}\delta\vect{x}_4,&\Rightarrow~& \delta\vect{x}_4(t)\equiv \vect{0},\Rightarrow \delta \vect{y}_4(t)\equiv \vect{0},\\
\delta\dvect{x}_5&=-\Dout_{55}\delta\vect{x}_5,~&\Rightarrow~& \delta\vect{x}_5(t)\equiv \vect{0},\Rightarrow \delta \vect{y}_5(t)\equiv \vect{0},
\end{alignat}
\end{subequations}
for $t\in\real_{\geq0}$. Substituting for $\vect{x}_3$ ans $\delta\vect{f}_2$ in~\eqref{eq::obfuscation_dyno_proof-x2}, 
we obtain
\begin{align}\label{eq::solution_confirm-2}
   \delta\dvect{x}_2&=-\Dout_{22}\delta\vect{x}_2+\vectsf{A}_{23}\textup{e}^{-\lL_{33}t}\delta\vect{x}_3(0)-\vectsf{A}_{23}\textup{e}^{-\lL_{33}t}\delta\vect{x}_3(0)\nonumber\\
   &=-\Dout_{22}\delta\vect{x}_2, ~\Rightarrow  \delta\vect{x}_2(t)=\textup{e}^{-\Dout_{22}t}\delta\vect{x}_2(0),
\end{align}
for $t\in\real_{\geq0}$. Finally using $\delta \vect{g}_2$ in~\eqref{eq::delta_g2_f2}, we 
\begin{align}\label{eq::out_2_confirm}
    \delta\vect{y}_2(t)&=\delta\vect{x}_2+\delta \vect{g}_2\nonumber\\
    &=\textup{e}^{-\Dout_{22}t}\delta\vect{x}_2(0)-\textup{e}^{-\Dout_{22}t}\delta\vect{x}_2(0)\equiv\vect{0},
\end{align}
for $t\in\real_{\geq0}$.
\end{proof}

\begin{proof}[Proof of Lemma~\ref{lem:case2_priv}]
If agent $1$ knows $\beta^i$, the proof follows from Lemma~\ref{lem::easy-prey-external-ad}. If agent $1$ does not know $\beta^i$, since it knows~\eqref{eq::nec-suf-admin-sig-a}, there exists at least one other agent $k\in\VV\backslash\{1,i\}$ whose $\beta^k$ is not known to agent $1$. We note that at the best case, $\beta^i+\beta^k$ can be known to agent $1$. Now consider $\beta_{ik}\in\real\backslash\{0\}$ and let $\beta^{i'}=\beta^i+\beta_{ik}$ and $\beta^{k'}=\beta^k- \beta_{ik}$,  and $\beta^{l'}=\beta^l$ for $l\in\VV\backslash\{i,k\}$. 
 Now consider an alternative implementation of algorithm~\eqref{eq::consensus-modified-x}-\eqref{eq::consensus-modified-y} with initial conditions $x^{l'}(0)=x^{l}(0)$ for $l\in\VV\backslash\{i,k\}$, $x^{i'}(0)=x^{i}(0)-\beta_{ik}$ and $x^{k'}(0)=x^{k}(0)+\beta_{ik}$ and obfuscation signals $f^{l'}(t)=f^{l}(t)$, $g^{l'}(t)=g^{l}(t)$ for $l\in\VV\backslash\{i,k\}$, $f^{i'}(t) = f^i(t) + d\,\beta_{ik} \text{e}^{-(\text{d}^i_{\text{out}}+ d)t}$, $g^{i'}(t) = g^i(t) + \beta_{ik} \text{e}^{-(\text{d}^i_{\text{out}}+d)t}$ and $f^{k'}(t) = f^k(t) - d\,\beta_{ik} \text{e}^{-(\text{d}^k_{\text{out}}+ d)t}$, $g^{k'}(t) = g^k(t) - \beta_{ik} \text{e}^{-(\text{d}^k_{\text{out}}+d)t}$, where $d\in\real$ is chosen such that $d>\max\{\text{d}^i_{\text{out}},\text{d}^k_{\text{out}}\}$. Let $t\mapsto {x^{l'}}(t)$ and $t\mapsto {y^{l'}}(t)$, $t\in\real_{\geq0}$, respectively, be the state and the transmitted signal of agent $l\in\VV$ in this alternative case. 
  We note that using  $\lim_{t\to\infty}\int_{0}^{t}d\beta_{ik}\text{e}^{-(\text{d}^i_{\text{out}}+d)\tau}\text{d}\tau=\frac{d\beta_{ik}}{\text{d}^i_{\text{out}}+d}$ and $\lim_{t\to\infty}\int_{0}^{t}d\beta_{ik}\text{e}^{-(\text{d}^i_{\text{out}}+d)\tau}\text{d}\tau=\frac{1}{\text{d}^i_{\text{out}}+d}$ we can show $\lim_{t\to\infty}\int_{0}^{t}(f^{l'}(\tau)\!+\!\dout^l\, g^{l'}(\tau))\,\textup{d}\tau=\beta^{l'}$, and  $\lim_{t\to\infty}\int_{0}^{t} \textup{e}^{-(t-\tau)}{g}^{l'}(\tau)\textup{d}\tau=\alpha$ for $l\in\VV$. Therefore, since  $\sum_{l=j}^N\beta^{j'}=0$, by virtue of Theorem~\ref{thm::main} we get
\begin{align}
&\lim_{t\to\infty}{x^l}'(t)=\frac{1}{N}\sum\nolimits_{j=1}^N\!\!{x^l}'(0)=\frac{1}{N}\sum\nolimits_{j=1}^N\!\!\mathsf{r}^l,~~ l\in\VV.\label{eq::init-alt-converge_beta_amb}
\end{align}
Next, let $\delta x^l(t)=x^l(t)-x^{l'}(t)$ and $\delta y^l(t)=y^l(t)-y^{l'}(t)$, $l\in\VV$.  Then,
\begin{subequations}\label{eq::diff_ik}
\begin{align}
&\begin{cases}
    \delta \dot{x}^l(t)=-\dout^l \delta {x}^l(t)+\!\!\sum\limits_{j=1}^N\mathsf{a}_{lj}\delta y^j(t),&l\!\in\!\VV\backslash\{i,k\},\\
   \delta \dot{x}^l(t)=-\dout^l \delta {x}^l(t)\!+\!\!\sum\limits_{j=1}^N\!\!\mathsf{a}_{lj}\delta y^j(t)\!+\!f^l\!-\!f^{l'},& l\!\in\!\{i,k\},
    \end{cases} \label{eq::diff_ik-a}\\
  &  \begin{cases}
     \delta y^l(t)\!=\delta x^l,&l\!\in\!\VV\backslash\{i,k\},\\
     \delta y^l(t)\!=\delta x^l+g^l\!-\!g^{l'},&l\!\in\!\{i,k\}. \end{cases} \label{eq::diff_ik-b}
\end{align}
\end{subequations}
To complete our proof, we want to show that $y^l(t)=y^{l'}(t)$~(or equivalently $\delta y^l(t)\equiv 0$), $l\!\in\!\VV$, for  $t\!\in\!\real_{\geq0}$, thus agent~$1$ cannot distinguish between the initial conditions $x^i(0)$ and $x^{i'}(0)$. Since, for a given initial condition and integrable~external inputs the solution of an ordinary differential equation is unique, we achieve this goal by showing that if $\delta y^l(t)=0$, $l\in\VV$ applied in the state dynamics~\eqref{eq::diff_ik-a}, the resulted output~\eqref{eq::diff_ik-a} satisfy $\delta y^l(t)\!\equiv\!0$, $l\!\in\!\VV$, $t\!\in\!\real_{\geq0}$. For this,~first note that since $\delta x^l (0)=0$ for $l\!\in\!\VV\backslash\{i,k\}$, then it follows from~\eqref{eq::diff_ik-a} with $\delta y^l(t)=0$, $l\!\in\!\VV$, that $\delta x^l(t)\equiv 0$. Subsequently, from~\eqref{eq::diff_ik-b}, we get the desired $\delta y^l(t)\equiv 0$, $t\!\in\!\real_{\geq0}$ for $l\!\in\!\VV\backslash\{i,k\}$. Next, we note that, from~\eqref{eq::diff_ik-a} with $\delta y^l(t)=0$, $l\!\in\!\VV$, given~$\delta x^i(0)=\beta_{ik}$ and $\delta x^k(0)=-\beta_{ik}$ we~obtain
\begin{align*}
    \delta x^i(t)=&\beta_{ik}\textup{e}^{-\dout^i t}
    -\beta_{ik}\textup{e}^{-\dout^i t}+\beta_{ik}\textup{e}^{-(\dout^i+d) t}\\
      =&\beta_{ik}\textup{e}^{-(\dout^i+d) t}\\
     \delta x^k(t)=&-\beta_{ik}\textup{e}^{-\dout^k t} +\beta_{ik}\textup{e}^{-\dout^k t}-\beta_{ik}\textup{e}^{-(\dout^k+d) t}\\
     =&-\beta_{ik}\textup{e}^{-(\dout^k+d) t}
\end{align*}
Subsequently, since $g^i-g^{i'}=-\beta_{ik}\textup{e}^{-(\dout^i+d) \tau}$ and $g^k-g^{k'}=\beta_{ik}\textup{e}^{-(\dout^k+d) \tau}$, from~\eqref{eq::diff_ik-b}, we get the desired $\delta y^l(t)\equiv 0$, $t\in\real_{\geq0}$ for $l\in\{i,k\}$, which completes our proof.
\end{proof}



\end{document}